\documentclass{article}
\usepackage[T1]{fontenc}
\usepackage[latin9]{inputenc}
\usepackage{booktabs}
\usepackage{url}
\usepackage{graphicx}
\usepackage{multirow}
\usepackage{color}
\usepackage{xcolor}

\usepackage{a4wide}
\usepackage{authblk}
\usepackage[numbers, sort]{natbib}
\usepackage{enumerate}
\usepackage{fancyhdr}

\fancypagestyle{plain}{%
    \fancyhf{}
    \fancyhead[L]{\sffamily\itshape Low-rank pre-smoothing}
}

\usepackage{amssymb,amsmath,amsthm}

\newtheorem{definition}{Definition}
\newtheorem{proposition}{Proposition}
\newtheorem{lemma}{Lemma}
\newtheorem{theorem}{Theorem}

\newtheorem*{unntheorem}{Theorem}
\newtheorem*{unncoro}{Corollary}

\newcommand{\vect}{\mathrm{vec}}
\newcommand{\tr}{\mathrm{Tr}}
\newcommand{\MSE}{\mathrm{MSE}}
\newcommand{\plim}{\mathrm{plim}}

\begin{document}

\title{Large multi-response linear regression estimation based on low-rank pre-smoothing}
\author[1]{Xinle Tian} 
\author[2]{Alex Gibberd} 
\author[1]{Matthew A. Nunes\thanks{\noindent Corresponding author: {\tt man54@bath.ac.uk}.\\ Author email addresses: {\tt xt373@bath.ac.uk} (Xinle Tian), {\tt a.gibberd@lancaster.ac.uk} (Alex
Gibberd), {\tt man54@bath.ac.uk} (Matthew A. Nunes), {\tt sr2081@bath.ac.uk} (Sandipan Roy).}} 
\author[1]{Sandipan Roy} 
 \affil[1]{Department of Mathematical Sciences, University of Bath, Bath, UK}
 \affil[2]{School of Mathematical Sciences, Lancaster University, Lancaster, UK}

\date{Published in \textit{Statistical Analysis and Data Mining}, \textbf{19}(2), e70072 (2026)\\
https://doi.org/10.1002/sam.70072}

\maketitle

\begin{abstract} 
Pre-smoothing is a technique aimed at increasing the signal-to-noise ratio in data to improve subsequent estimation and model selection in regression problems. However, pre-smoothing has thus far been limited to the univariate response regression setting. 
However, there are many scientific applications in which interest lies in multi-response regression problems, particularly when the number of responses is large.
Motivated by this setting, this article proposes a technique for data pre-smoothing based on low-rank approximation. We establish theoretical results on the performance of the proposed methodology, which show that in this large-response setting, the proposed technique outperforms ordinary least squares estimation with the mean squared error criterion, whilst being computationally more efficient than alternative approaches such as reduced rank regression. We quantify our estimator's benefit empirically in a number of simulated experiments. We also demonstrate our proposed low-rank pre-smoothing technique on real data arising from the environmental and biological sciences. \\

\bigskip

\noindent {\bf {\em Keywords}}: low-rank approximation; multi-response model; parameter estimation; prediction; pre-smoothing

\end{abstract}

\section{Introduction}

In multivariate data analysis, a common task in analysing multi-response data is to model the relationships between the vector of responses of interest and a number of explanatory variables. Such data are observed in a multitude of scientific applications, including ecology \citep{anderson1997effects}, genetics \citep{sivakumaran2011abundant, liu2014review}, chemometrics \citep{skagerberg1992multivariate, burnham1999latent}, bioinformatics \citep{kim2009multivariate}, population health \citep{bonnini2022relationship}, as well as economics and finance \citep{gospodinov2017spurious, petrella2019joint}.  Whilst component-wise analysis may be an option, analyzing each component separately fails to consider the dependence between the components of the responses; this can result in estimators of model parameters that are less efficient than those that account for potential correlation between responses (see e.g. \cite{srivastava2003predicting, izenman2008modern}), and thus multivariate linear models are often more appropriate. See \cite{anderson1958introduction} for a comprehensive introduction to classical parametric models. Parameter estimation in classical multi-response regression settings is achieved with the ordinary least squares (OLS) estimator, with the Gauss-Markov theorem showing this estimator to be the best linear unbiased estimator (BLUE) \citep[Chapter 7]{johnson07:applied}. 

In many analysis settings, the mean squared error (MSE) criterion is commonly used to measure the quality of a prediction technique, incorporating the well-known bias-variance trade-off in assessing estimators.  One way to achieve a smaller MSE in univariate regression settings is to use so-called \textit{pre-smoothing}, in which response data is smoothed in some way prior to estimation via OLS \citep{faraldo1987efficiency, cristobal1987class, janssen2001efficiency}. The intuition behind this approach is that this first step increases the signal-to-noise ratio in the data, thus reducing the variance contribution in the bias-variance trade-off and consequently reducing parameter uncertainty. Other authors have also noted that pre-smoothing can help with model selection in uni-response regression \citep{aerts2010model}.  The pre-smoothing paradigm has subsequently been used in other modelling settings, such as with functional data \citep{hitchcock2006improved, hitchcock2007effect, ferraty2012presmoothing}, as well as semi-parametric and survival models \citep{musta2022presmoothing, tedesco2025instrumental,  musta2024regression}. Despite these contributions, pre-smoothing has not been explored in the multi-response regression setting.  The work in this article is particularly motivated by settings in which the number of responses is large, a situation commonly found in applications such as genomics \citep{wille2004sparse} and spatiotemporal modelling \citep{liang2015assessing}.

To this end, this article proposes an alternative to the OLS estimator for multi-response regression models using data pre-smoothing, in the spirit of \cite{cristobal1987class, janssen2001efficiency} in the one-dimensional setting.  In what follows, we illustrate this methodology with a low-rank approximation to achieve the pre-smoothing, hence we term our proposed technique \textit{low-rank pre-smoothing} (LRPS). Our proposed low-rank pre-smoothing technique is naturally ``model-free'' in the sense that it does not need additional covariate information such as the nonparametric pre-smoothers proposed in the one-dimensional setting, see also the examples of \citet[Section~2.3]{aerts2010model}.  We investigate how LRPS can be used in the linear regression setting, and demonstrate that it achieves a lower MSE than the OLS estimator and shows particular improvements when the number of responses is large, or the signal-to-noise ratio is small. Our technique is computationally efficient since estimation can be performed directly in the latent space defined by the low-rank approximation. The proposed methodology can be seen as an alternative to the popular reduced rank regression (RRR) technique for such settings, and complements other multi-response regression methods in the literature such as principal component regression (PCR) and partial least squares regression (PLS).

This article is structured as follows. In Section~\ref{sec:LRPS}, we provide a detailed description of our low-rank pre-smoothing estimator. In particular, the theoretical properties of the estimator are explored in Section~\ref{sec:property}. We then present the results of the simulation studies in Section~\ref{sec:Simulation-Study} to examine the effectiveness of the proposed LRPS estimator. We illustrate the practical use of our proposed method in prediction tasks of real data arising in biology and environmental science in Section~\ref{sec:data}.  Finally, we make some concluding remarks in Section \ref{sec:concs}. 

\section{Low-rank Pre-smoothing in multi-response regression}\label{sec:LRPS}

Consider the multivariate multiple outcome linear regression model
\begin{equation}
Y=XB+E\;,\label{eq:mrreg}
\end{equation}
where $Y=(y_{1},\ldots,y_{n})^{\top}$ and $E=(e_{1},\ldots,e_{n})^{\top}$
are $n\times q$ dimensional random matrices representing the outcomes
and errors respectively. We study the fixed design case whereby $X\in\mathbb{R}^{n\times p}$,
and $B\in\mathbb{R}^{p\times q}$ are parameters to be estimated.
We further assume that $e_{i}$ is independent from $e_{j}$ for all
$j\ne i$, where $i,j = 1,2,\dots,n$, so rows are assumed to be independent and identical draws from a distribution $G$ and have constant covariance $\mathrm{Cov}(e_{i})=\Sigma_{e}$. In the case where $e_{i}\overset{iid}{\sim}N_{q}(0,\Sigma_{e})$,
the model (\ref{eq:mrreg}) can be alternatively written as
\begin{equation}
\vect(Y)\sim\mathcal{N}_{nq}\left(\vect(XB),\Sigma_{e}\otimes I_{p}\right).\label{eq:mrregvec}
\end{equation}
When $n>p$ the maximum-likelihood estimator for $B$ under \eqref{eq:mrregvec}
is well defined, and is given by the OLS solution $\hat{B}=S_{X}^{-1}S_{XY}$
where $S_{X}=n^{-1}X^{\top}X$ and $S_{XY}=n^{-1}X^{\top}Y$. This
estimator has the distributional property (see e.g., \citet[Chapter 6]{mardia95}) that
\begin{equation}\label{eq:olsdist}
\vect(\hat{B})\sim\mathcal{N}_{pq}(\vect(B),n^{-1}\Sigma_{e}\otimes S_{X}^{-1})\;,
\end{equation}

and as a consequence, the estimator is consistent and unbiased \citep[Chapter 4]{greene2018econometric}.

Whilst asymptotically the OLS maintains the best unbiased linear estimate
(BLUE) properties, the $pq$ dimensional quantity $\vect(B)$ is still relatively
high-dimensional. As such, performance in finite samples may result
in noisy estimates, with considerable variance. The aim of this paper
is to propose an estimator that enjoys superior finite-sample (and
asymptotic) performance when estimating $B$. Our proposal is straightforward,
and aims to reduce the proportion of the noise, i.e. variance of errors,
that contributes to the regression problem at hand. We now formally define our proposed estimator.

\begin{definition}{The LRPS estimator}

Let $S_{Y}=n^{-1}Y^{\top}Y$ and consider its eigendecomposition
$\hat{V}\hat{D}\hat{V}^{\top}=S_{Y}$. The \emph{Low-Rank Pre-Smoothing
(LRPS) estimator} for $B$ is then given by 
\begin{equation}
\tilde{B}=S_{X}^{-1}S_{XY}\hat{V}_{k}{\hat{V}_k}^{\top}\;,\label{eq:LRPS-1}
\end{equation}
where $\hat{V}=(\hat{V}_{k},\hat{V}_{\backslash k})$, with $\hat{V}_{k}=(\hat{v}_{1},\ldots,\hat{v}_{k})$
being the first $k\le q$ eigenvectors, and $\hat{V}_{\backslash k}$
the remaining $q-k$ eigenvectors. We assume the ordering of $k$ matches that of the
eigenvalues $\hat{D}=\mathrm{diag}(\hat{d}_{1},\ldots,\hat{d}_{q})$
and $\hat{d}_{k}\ge\hat{d}_{k+1}$. 

\end{definition}
The LRPS estimator can be seen as the OLS estimator applied to a ``pre-smoothed''
set of outcomes $\tilde{Y}=Y\hat{W}_{k}$, with $\hat{W}_{k}=\hat{V}_{k}{\hat{V}_k}^{\top}$
acting to project $Y$ onto the subspace spanned by $\hat{V}_{k}$,
and then ${\hat{V}_k}^{\top}$ projecting back into the original dimensions
of $Y$. The estimator is simple to compute, with the truncated eigenvectors
$\hat{V}_{k}$ found in $\mathcal{O}(kq^{2})$ computational time.
Furthermore, via the Eckhardt-Young theorem \citep{jolliffe2016principal}
we have that $\tilde{Y}$ is the best rank $k$ approximation of $Y$
under the least-squares (Frobenius) loss. The method is feasible as
$\hat{V}_{k}$ are the empirical eigenvectors associated with the
observed $Y$, and not on a population, or model-dependent quantity. Our motivation is
somewhat similar to that of \citet{cristobal1987class,janssen2001efficiency}
who propose non-parametric pre-smoothing of an outcome before performing
the regression. However, in our case, we have a multi-outcome (multivariate)
setting, and the smoothing is performed via a specific device, that
is the projection onto the most influential eigenvectors of the outcome.  We note here that other suitable multivariate non-parametric data smoothers could be used for pre-smoothing, however, they would need other (covariate) information as inputs to the smoother, and potentially suffer from the curse of dimensionality in high-dimensional multi-response data regimes.

At this point, we should note that a similar methodology can be employed
within the class of methods known as reduced rank-regression \citep{izenman1975reduced}.
We elaborate on the similarities and differences between LRPS and
these methods later in Section \ref{sec:relationtorrr}, however as the name suggests, reduced rank-regression
aims to find a low-rank approximation of $B$, whereas LRPS makes
no such explicit assumption, and indeed, we can set $k\ge\min(n,p)$
in LRPS to enable $\tilde{B}$ to remain full rank. A further
difference between the methods is that LRPS doesn't require knowledge
of $X$ prior to the choice of eigenvectors, as $V_{k}$ is obtained
directly from $Y$. However, the optimal choice of $k$ will in general
depend on the shape of the spectrum (eigenvalues) of $XB$, $E$,
and the dimensionality $p,n,q$. This aspect of our proposed method is explored further in Section \ref{sec:Simulation-Study}.

\subsection{Theoretical properties of the LRPS estimator}

\label{sec:property}

In this section, we discuss the asymptotic distribution for the LRPS
estimator and the potential efficiency gains relative to OLS. Our
analysis is in the standard dimensional setting where $p$ and $q$ are assumed
fixed with $p < n$, and we let $n\rightarrow\infty$. To study the estimator in
this setting, we consider the population covariance $\Sigma_{y}=n^{-1}(XB)^{\top}(XB)+\Sigma_{e}=B^{\top}S_{X}B+\Sigma_{e}$, its eigendecomposition $\Sigma_{y}=VDV^{\top}$, and the smoothing operator $W_{k}=V_{k}{V_k}^{\top}$ based on the
first $k$ eigenvectors. 

\begin{theorem}{Asymptotic distribution of LRPS}

Under model (\ref{eq:mrreg}), suppose the moments of the errors $\{e_{1},\ldots,e_{n}\}$
exist (and are finite) up to the fourth order. Assuming that $\lim_{n\rightarrow\infty}S_{X}=\Sigma_{X}$,
then we have
\begin{equation}
n^{1/2}\vect\left(\tilde{B}-BW_{k}\right)\stackrel{D}{\rightarrow}\mathcal{N}_{pq}\left(\mathbf{0},W_{k}\Sigma_{e}W_{k}\otimes\Sigma_{X}^{-1}\right)\ \textrm{as $n\rightarrow\infty$}.\label{eq:asydist}
\end{equation}
\label{theorem:lrpsdist}

\end{theorem}

The proof of the above can be found in Appendix \ref{sec:appendix}, and relies on
the consistency of the sample eigenvectors $\plim_{n\rightarrow\infty}\hat{V}_{k}=V_{k}$ 
\citep{neudecker1990asymptotic,bura2008distribution}, where $\plim$ denotes as convergence in probability; without loss of generality, we also assume that the design matrix $X$ (and data matrix $Y$) have been mean-centered. The result
shows $\tilde{B}$ is a root-$n$ consistent estimator of the corresponding
low-rank projection of $B$, $BW_{k}$; we note here that whilst motivated via the OLS estimator under the Gaussian assumption in \eqref{eq:mrregvec}, the distributional result of Theorem \ref{theorem:lrpsdist} holds with only mild moment assumptions on the error distribution $G$ (see Appendix \ref{sec:appendix}). Furthermore, this illustrates
that whilst the LRPS estimator may be a biased estimator of the coefficient
matrix $B$ (for $k<q$), it enjoys better mean squared error (MSE)
performance than the OLS in some important scenarios. 

Let 
$\MSE(\tilde{B}):=\mathbb{E}_{E}(\|\tilde{B}-B\|_{F}^{2})$
where $\| \cdot \|$ denotes as Frobenius norm. We take expectations over the error matrix $E$ with $n$ rows and $q$ columns,
e.g. as in the model (\ref{eq:mrreg}). For a large $n$, if the asymptotic
distribution of Theorem 1 holds, the MSE of LRPS can be approximated
as 
\begin{equation}
\MSE(\tilde{B})\approx\underset{\mathrm{Var}(\tilde{B})}{\underbrace{\frac{\bar{\alpha}_{X}}{n}\tr(\hat{W}_{k}\Sigma_{e})}}+\underset{\mathrm{Bias}^{2}(\tilde{B})}{\underbrace{\tr\{(I_{q}-\hat{W}_{k})B^{\top}B\}}},\label{eq:lrpsmse}
\end{equation}
\label{coro:mse} where $\bar{\alpha}_{X}:=\tr(S_{X}^{-1})$. The
first term represents the variance of the estimator, and the second,
the square of the bias\footnote{Note that our approximate statement here misuses the asymptotic properties
in the non-asymptotic setting. A more precise statement could potentially
be obtained via application of Berry-Esseen  type limit theorems on
the sample eigenvectors, however, this requires further technical
assumptions and we prefer to leave finite-sample inference for future
work.}.

Note that the form of the mean squared error expression in (\ref{eq:lrpsmse})
is, as expected, dependent on the chosen number of low rank components,
$k$. This recovers the MSE for the OLS when $k=q$, in which case
the bias disappears as $\hat{W}_{q}=I_{q}$. The variance reduction
relative to OLS is given by
\[
\frac{\mathrm{Var}(\tilde{B})}{\mathrm{Var}(\hat{B})}\approx\frac{\|\Sigma_{e}^{1/2}\hat{V}_{k}\|_{F}^{2}}{\|\Sigma_{e}^{1/2}\|_{F}^{2}}\le\frac{\sum_{i=1}^{k}\gamma_{i}^{2}(\Sigma_{e}^{1/2})}{\sum_{i=1}^{q}\gamma_{i}^{2}(\Sigma_{e}^{1/2})}\;,
\]
where $ \gamma_i(\Sigma_e^{1/2}) $ denotes the $ i $-th largest singular value of $\Sigma_e^{1/2}$. Thus if the spectrum of the errors, which refers to the spread of singular values of the error covariance matrix, is relatively flat, then the reduction
can be of order $k/q$, a substantial benefit when $q$ is large. On the other hand, if the spectrum of the errors is spiked, e.g. if
$\Sigma_{e}$ were geometrically decaying, then the variance reduction
benefits of LRPS may be minimal. In terms of bias, we see this is
related to the proportion of the signal (as expressed via $B$ through
$X$) projected into the orthocomplement
 of the space spanned by the first
$k$ singular vectors of $Y$. 
 The eigenvectors $\hat{V}_{k}$ will
in general not be optimal for representing the signal $XB$, however,
are chosen based on a combination of this signal, and the error $EV_{k}$. 

\subsection{Relation to reduced-rank-regression}\label{sec:relationtorrr}

In the case when $k\le\min(p,n)$ it is of interest to understand
the performance of LRPS in terms of estimation bias, in particular,
in comparison to reduced-rank-regression (RRR) methods \citep{izenman1975reduced,Reinsel1998}
which aims to solve the problem
\begin{equation*}
\bar{B} \in  \arg \min_{B\in\mathbb{R}^{p\times q}}\|Y-XB\|_{F}^{2}
 \qquad \mathrm{s.t.}\;r(B)\le k\;,
\end{equation*}

for a given $k\le\min(p,q)$, where $r(V)$ denotes the rank of $V$. A classical approach to this problem
is to minimise $B$ after projecting onto the space spanned by the
rank-$k$ approximation of $\hat{Y}=HY$ where $H=n^{-1}XS_{X}^{-1}X^{\top}$. 

\begin{definition}{Reduced-rank regression}\label{dfn:rrr}

Consider the decomposition:
\[
U_{k}G_{k}{U_{k}}^{\top}+U_{\backslash k}G_{\backslash k}U_{\backslash k}^{\top}=\frac{1}{n}Y^{\top}HHY=\hat{B}^{\top}S_{X}\hat{B}\;,
\]
where $U_{k},G_{k}$ are matrices representing the first $k$ eigenvectors
and eigenvalues of the decomposition respectively. The \emph{Reduced-Rank Regression (RRR)}
estimator is then defined as
\begin{equation}
\bar{B}=S_{X}^{-1}S_{XY}U_{k}{U_{k}}^{\top}\;.\label{eq:rrr_reinsel}
\end{equation}
\end{definition}

The form of the estimator follows from noting that when $n>p$ we
have 
\[
\bar{B}=\arg\min_{r(B)\le k_{*}}\{\|Y-X\hat{B}\|_{F}^{2}+\|XB-X\hat{B}\|_{F}^{2}\}\;.
\]
Applying Theorem 2.1 in \citet{Reinsel1998} to optimize over the second term in the equation above leads to (\ref{eq:rrr_reinsel}). We thus see that LRPS and RRR are
closely related, the difference being the choice of basis used
in the projection step. When $n>p$, then RRR represents the optimal
solution in the Gaussian error setting (where it attains the MLE and
thus BLUE property), {\textit when the low-rank $k_{*}=r(B)\le\min(q,p)$
assumption holds}. Importantly, our proposed LRPS method can also be applied when $k\ge p$,
whereas RRR falls back to the OLS solution (for $k\ge\min(p,q)$ we
obtain $\bar{B}=\hat{B}U_{k}{U_{k}}^{\top}=\hat{B}$). In this setting,
RRR and OLS may have high-variance, whilst LRPS can still benefit
from variance reduction. 

\subsection{Other Related work}

As well as RRR, it is also natural to draw a connection to principal component regression (PCR) \cite{massy1965principal, jolliffe2011principal}, which likewise relies on low-dimensional representations using a least squares objective. In PCR, the predictors are first projected onto the leading principal components of the design matrix $X$, and the response is then regressed onto this reduced set of components. The resulting estimator thus depends on the principal subspace of $X$, rather than directly on the response as in LRPS. Table \ref{table:methods} provides a high-level overview of the three methods, LRPS, RRR and PCR in comparison to OLS. 

\begin{table}[!h]
\centering

\begin{tabular}{l p{8.5cm}}
\hline
\textbf{Method} & \textbf{Coefficient Estimator} \\
\hline
\\
\textbf{OLS} 
& 
$\displaystyle 
\hat B_{\text{OLS}} = 
(X^\top X)^{-1} X^\top Y
$
\\[1.0em]

\textbf{RRR} 
&
$\displaystyle 
\hat B_{\text{RRR}} 
= 
\hat B_{\text{OLS}} 
U_k U_k^\top,
$
\\
[1.0em]

\textbf{LRPS} 
&
$\displaystyle
\hat B_{\text{LRPS}} 
= \hat B_{\text{OLS}} 
V_k V_k^\top,
$
\\[1.0em]
\textbf{PCR} 
&
$\displaystyle 
\hat B_{\text{PCR}}
= (Z_k^\top Z_k)^{-1} Z_k^\top Y\quad\mathrm{where}\quad Z_k = X Q_k 
$
\\[1.0em]

\hline
\end{tabular}
\caption{Comparison of OLS, RRR, LRPS and PCR, where we define $U_k,V_k,Q_k$ be the top $k$ right singular vectors of $X\hat{B}_{\mathrm{OLS}},Y$ and $X$ respectively. \label{table:methods}}
\end{table}

One may also consider the James-Stein (JS) shrinkage estimator \cite{james1961estimation}, which shrinks the OLS estimator toward zero through a scale-dependent shrinkage factor. Unlike approaches such as LRPS or RRR, the JS estimator does not explicitly impose a low-rank structure on either the response matrix $Y$ or the covariates $X$, but instead achieves regularisation solely via global shrinkage of the OLS coefficients.  We provide empirical comparison of the methods in Table \ref{table:methods} in Section \ref{sec:Simulation-Study}.

\subsection{Computational considerations}\label{sec:complexity}

In this section, we consider the computational complexity of the LRPS estimator, and compare it to that of RRR, to motivate the regression situations in which LRPS will provide benefits.  

\begin{proposition}{Computational cost of LRPS \label{prop:complexity}
}

For the regression model (\ref{eq:mrreg}), the computational complexity of the LRPS estimator is\footnote{ The power of the cubic term ($p^3$) in the expression for the complexity can be reduced to $p^\gamma$, for some $\gamma<3$, depending on the algorithm used to compute the inverse of $S_X$ (see e.g. \cite{tonks}), but this is not crucial to our discussion here.} 
\[
O(np^2 + p^3 + k((n+p)(p+q)+\min\{p,q\})).
\]
\end{proposition}
The details of the computation of the above expression can be found in Appendix \ref{sec:appendixprop}.  In essence, the LRPS estimator can be computed directly in the latent space, and thus involves matrix multiplications of lower dimension to that of RRR. 

\paragraph{Comparison with RRR.}  Using similar arguments to those for the derivation of the result in Proposition \ref{prop:complexity} (Appendix \ref{sec:appendixprop}), the complexity of RRR using equation \eqref{eq:rrr_reinsel} is
$$O(np^2+p^3) + O(nqk) + O(npq) + O(p^2q) + O(npq) + O(q^2k) + O(pq^2).$$

Both LRPS and RRR algorithms perform two common operations, namely computing the inverse of $S_X$ of complexity $O(np^2+ p^3)$, and performing the truncated SVD of a matrix of dimension $n\times q$ of complexity $O(nqk)$ ($Y$ for LRPS, $X\hat{B}$ for RRR).  These are the first two terms in the computation above and in Appendix \ref{sec:appendixprop} for the LRPS estimator.  Removing this from their respective computations, the terms corresponding to the two estimators can be more directly compared:
$$
\begin{array}{cccccc}
RRR:&  O(npq)    & O(p^2q) & O(pq^2)& O(q^2k)& O(npq) \\
&  \updownarrow & \updownarrow & \updownarrow & \updownarrow &  \\
LRPS:&    O(npk) & O(p^2k) &O(pqk)& O(\min\{p,q\}k)&
\end{array}
$$


For the first three terms, there is a scaling factor (reduction) of $\frac{k}{q}$ from RRR to LRPS, a significant reduction for the fourth term, and the fifth term of RRR is additional to the computation.  In light of this, LRPS is expected to be more computationally efficient than RRR if $k$ is chosen to be low (due to the first three terms), and for large $q$ settings as signified by the fourth term.

\subsection{Empirical analysis and finite-sample behaviour}

\label{subsec:Empirical-Analysis-and}

To give some intuition on the differences between the three estimators, LRPS, RRR, and
OLS, we consider a simple simulation study in which we empirically track
the MSE, the bias, and the variance of the estimators as a function
of $k$. We consider two scenarios, one where the true rank $r(B)=k_{*}=p/2$
is small relative to $\min(p,q)$, and the other where $k_{*}=\min(p,q)$.
For simplicity, we let $p=q=10$, $k_{*}=\{5,10\}$ and consider three
sample sizes $n=\{20,30,50\}$ to assess the finite-sample behaviour
of the estimators. In each case, we assume the model (\ref{eq:mrregvec}),
we simulate $X$ and $E$ as a set of isotropic Gaussian random vectors
(independent rows). Further details on the simulation design can be
found in Section \ref{sec:Simulation-Study}, where we present results
from a more comprehensive study.

\begin{figure}[!h]
\includegraphics[width=1\columnwidth]{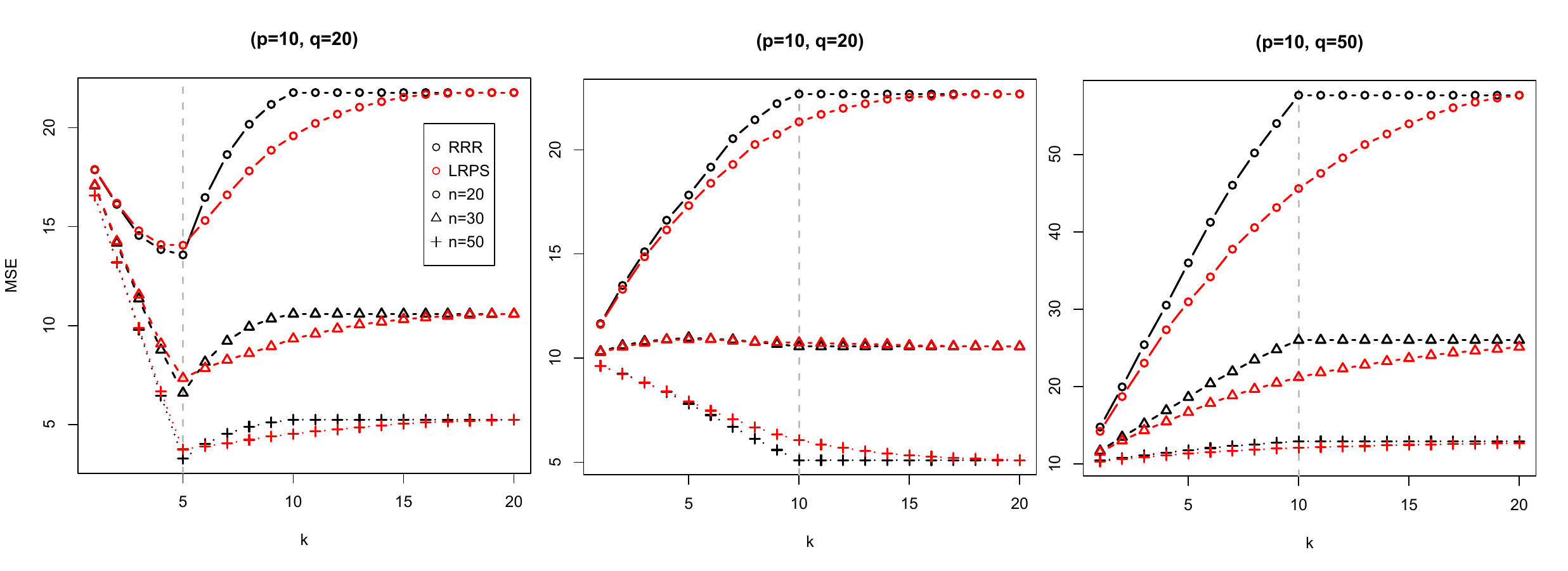}

\caption{Performance of LRPS and RRR in terms of empirical $\mathrm{\hat{MSE}=}m^{-1}\sum_{s=1}^{m}\|Y^{(s)}-X^{(s)}B^{(s)}\|_{F}^{2}$
where $Y^{(s)},X^{(s)}$ and $B^{(s)}$ represent the data, and estimate
(LRPS in red, and RRR in black) for simulation $s$, and $m=100$
is the number of simulations.\label{fig:Performance-of-LRPS}}
\end{figure}

The results in Figure \ref{fig:Performance-of-LRPS} shed some light
on the similarities and differences between LRPS and RRR. When the true model is
low-rank, e.g. when $k_{*}=5$, we see that for all sample sizes
the minimum MSE is attained at $r(\bar{B})=k_{\star}$
and RRR performs slightly better in terms of estimation error at this
minimum. On the other hand, when the low-rank assumption no longer
holds ($k_{*}=10$), we see that LRPS outperforms RRR in the small sample size
setting, whereas RRR becomes superior as the sample size $n$ grows. Indeed, in the small
sample setting, a rank-1 approximation is preferred
by both approaches, however, LRPS still maintains better performance.
In the scenario where $q$ is large relative to $p$ (Figure \ref{fig:Performance-of-LRPS}, right panel), LRPS exhibits superior estimation performance in the full-rank scenario, consistently over a range of sample sizes. This latter point is important as the setting where $q>p$ often occurs in practice. For instance, this is often observed in datasets from air pollution monitoring stations, where measurements of several pollutants are collected across multiple locations. This is discussed later in our application of LRPS to real data.

\begin{figure}[!h]
\begin{centering}
\includegraphics[width=0.9\columnwidth]{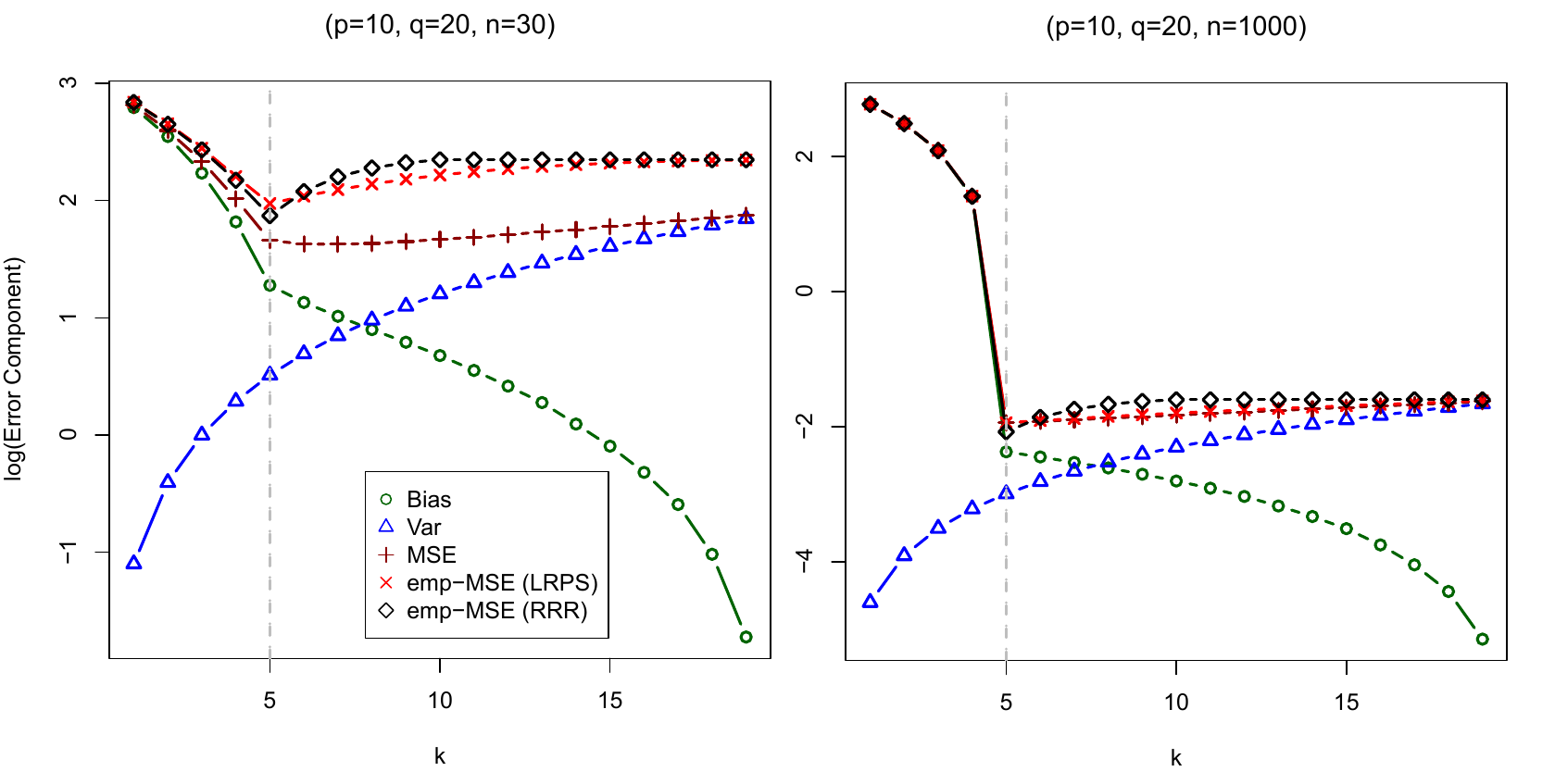}
\par\end{centering}
\caption{The $\mathrm{bias}^{2}(\tilde{B})$, $\mathrm{Var}(\tilde{B})$,
and $\mathrm{MSE}(\tilde{B})$ values extrapolated from the LRPS
asymptotic distribution in Theorem \ref{theorem:lrpsdist}, compared
to the finite-sample (empirical) MSE for both LRPS and RRR. In this
simulation, $k_{*}=5$ and other than the sample size, the conditions
are the same as those given in the left panel of Figure \ref{fig:Performance-of-LRPS}.
\label{fig:biasvar}}
\end{figure}

A characteristic property of LRPS is that the estimator responds
to changes in $k_{*}>\min(p,q)$, whereas RRR falls back to the OLS estimator
in this regime. To investigate the bias-variance tradeoff occurring
due to the pre-smoothing we plot the $\mathrm{bias}^{2}(\tilde{B})$,
$\mathrm{Var}(\tilde{B})$, and $\mathrm{MSE}(\tilde{B})$
components from (\ref{eq:lrpsmse}), and compare them with the empirical
MSE achieved in a finite sample, in this case, letting $n=30$ and
$n=1000$. The results are presented in Figure \ref{fig:biasvar}. As one may expect, extrapolating from the asymptotic
to the finite sample setting may not be appropriate for small samples,
but the approximation becomes tight as $n$ increases. Comparing the dark red (+) and red ($\times$) lines in Figure \ref{fig:biasvar}, we note that they essentially overlap in the larger $n=1000$
setting. We now provide more insight into the performance of LRPS through in-depth simulations.

\section{Simulation study}

\label{sec:Simulation-Study} In this section, we evaluate the performance of our proposed
technique by using simulated examples to assess the estimation of the regression
coefficient matrix, as well as out-of-sample predictive performance. The data $Y, X$ are simulated according to \eqref{eq:mrreg}, where we fix the design of the covariates to be $X\sim\mathcal{N}_{n\times p}(0,I_p\otimes I_n)$. We vary the specification of the true coefficient matrix $B$, and
the covariance structure of the errors $\Sigma_{e}$. The choices
here cover a range of scenarios that are of interest in practice and
help further elucidate the behaviour of LRPS (and RRR) when the spectrum
(eigenstructure) of the signal $XB$ and noise $E$ vary. The various conditions for $B$ and $\Sigma_{e}$ are given below.

\subsubsection*{Coefficient matrix conditions: }

\begin{description}

\item{\emph{1 - Sparse and low-rank:}} Let $s$ be the expected
rank of the coefficient matrix. We first select active elements for
$B'_{ij}$ by sampling from the Bernoulli distribution, $B'_{ij}\sim\text{Bernoulli}(s^{-1})$
for $j=1,\ldots,s$. The non-zero coefficients in each row are then
scaled (by dividing by the row-average) to maintain $\|B_{i,\cdot}\|_{2}=1$. 

\item{\emph{2 - Dense and unknown rank:}} Two random
matrices $A_{L} \in \mathbb{R}^{p \times p}, A_{R} \in \mathbb{R}^{q \times q}$ are generated with standard normal, and $\text{Uniform}(0,1)$
entries respectively. Let $L$ and $R$ be the left and right singular
vectors of $A_{L}$ and $A_{R}$, the coefficient matrix is then
generated according to
\[
B=A_{L}\mathrm{diag}(d_{1},\ldots d_{\min(p,q)})A_{R}^{T}\;,
\]
where $d_{i}=\lambda^{i}$ exhibits geometric decay. 

\item{\emph{3 - Dense and known rank:}} The matrix $B$ is constructed
as in Condition 2, however, the diagonal singular values are replaced
with $d_{i}=\lambda$ for $i=1,\ldots,k_{*}$ and $d_{i}=0$ for $k_{*}+1,\ldots,\min(p,q)$.
This creates a matrix which has rank $k_{*}$. Note, the examples
demonstrated in Section \ref{subsec:Empirical-Analysis-and} also
use matrices generated in this manner.

\end{description}

\newpage

\subsubsection*{Noise conditions:}

\begin{description}

\item{\emph{1 - Isotropic IID noise:}} The covariance of the errors
is given by $\Sigma_{e}=\sigma^{2}I_{q}$.

\item{\emph{2 - Geometric Decaying Toeplitz:}} We let $\Sigma_{e}=\mathrm{Toeplitz}(\rho)$
such that $\Sigma_{e;i,j}=\sigma^{2}\rho^{|i-j|}$.

\end{description}

Note that simulations in this section as well as the examples in Section \ref{subsec:Empirical-Analysis-and} were produced in the {\em R} statistical computing environment \citep{Rcore}. 

\subsection{Estimation error results}

This section examines the estimation error incurred by LRPS and RRR
in the scenarios described above. As in the previous
section, we focus on understanding the performance of the methods
over a range of $k$, and examine how this performance is dependent
on the composition of the signal and noise. For simplicity, we fix
$n=100$ and report performance as measured by the empirical mean squared error from the true coefficient matrix $B$, averaged across $m=100$ experiments. 

\subsubsection*{Large $q$, moderate $p$}

In this case, we consider when $q=100>p=10$ and thus the range of
$k$ with which we can apply LRPS is extended relative to RRR (which
we recall limits $k\le\min(p,q)$). In this scenario, we should expect
that the structure of the error covariance plays a key role, and
given the number of outcomes, the regression problem inherits a lower
signal-to-noise as $\|\Sigma_{e}\|_{F}\asymp q$ whilst $\|XB\|_{F}\asymp\sqrt{pq}$.
We run three experiments in this setting: 
\begin{itemize}
\item The first probes what we may expect to happen if $B$ were sparse,
under $B$ condition 1. In this case, the rank of the matrix is bounded
by $s$, which we choose to be $p/2=5$, since sparsity is a common
assumption in high-dimensional regression, we wanted to investigate
how methods that previously would focus explicitly on a low-rankness
assumption can also be applied in this setting. 
\item Second, we investigate the case where $B$ is dense, but with decaying
singular values (Condition 2), and we consider $\lambda=\{0.9,0.5,0.2\}$
to study the impact of the decay rate. 
\item Third, we examine the impact of a spiked error covariance via a geometrically
decaying correlation structure with values $\rho=\{0.9,0.5,0.2\}$.
The specification of $B$ is given via Condition 3 with a flat spectrum ($\lambda=1$).
\end{itemize}
The last two experiments fix $\sigma^{2}=1$, whereas the first fixes
the design of $B$ and varies $\sigma^{2}=\{0.5,1,2\}$ to directly
investigate the impact of the error variance on estimation performance.
Figure \ref{fig:singularvals} summarises the distribution of $\gamma_{j}^{2}(B)$ and $\gamma_{j}(\Sigma_{e})$,
where for each $j=1,\ldots q$ we take the average over $m=100$ experiments.

\begin{figure}[!h]
\includegraphics[width=1\columnwidth]{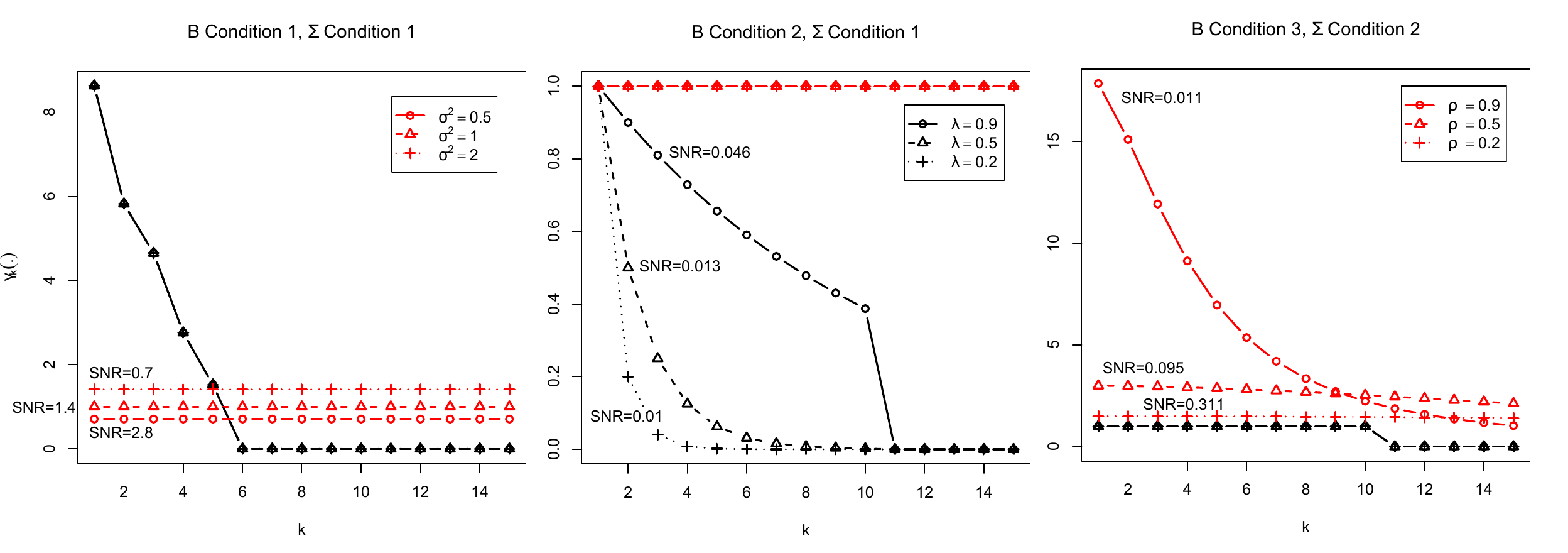}

\caption{Distribution of singular values for the signal $\gamma_{j}(B)$ and
noise $\gamma_{j}^{1/2}(\Sigma_{e})$, in black and red respectively.
The signal-to-noise ratio $\mathrm{SNR}:=\|XB\|_{F}/\|\Sigma_{e}\|_{F}$
is given next to each line, and the lines represent the average value
of the spectrum over $m=100$ simulations. \label{fig:singularvals}}
\end{figure}

The results of the experiment are presented in Figure \ref{fig:res_large_q}, with experiments
1--3 running from left to right. The figure shows a range of empirical
behaviours for the RRR and LRPS estimators and illustrates both estimators'
potential benefits beyond the OLS estimation. The best rank approximation,
i.e. the $k$ which minimises the estimation error is a function of
both the specification of the signal (via $B$), and the noise structure---in
general, a higher signal-to-noise ratio mandates the choice of a smaller
$k$. Intuitively, this observation makes sense, as in the high-noise settings,
we should desire to constrain the estimator more to avoid excessive
variance, i.e. the variance reduction benefits of the low-rank projection
dominate the error. Observe, at least when $q>p$, that LRPS is often
superior to RRR in these approximations, e.g. when comparing the rank-1
approximation performance in the middle and right plots.

\begin{figure}[!h]
\includegraphics[width=1\columnwidth]{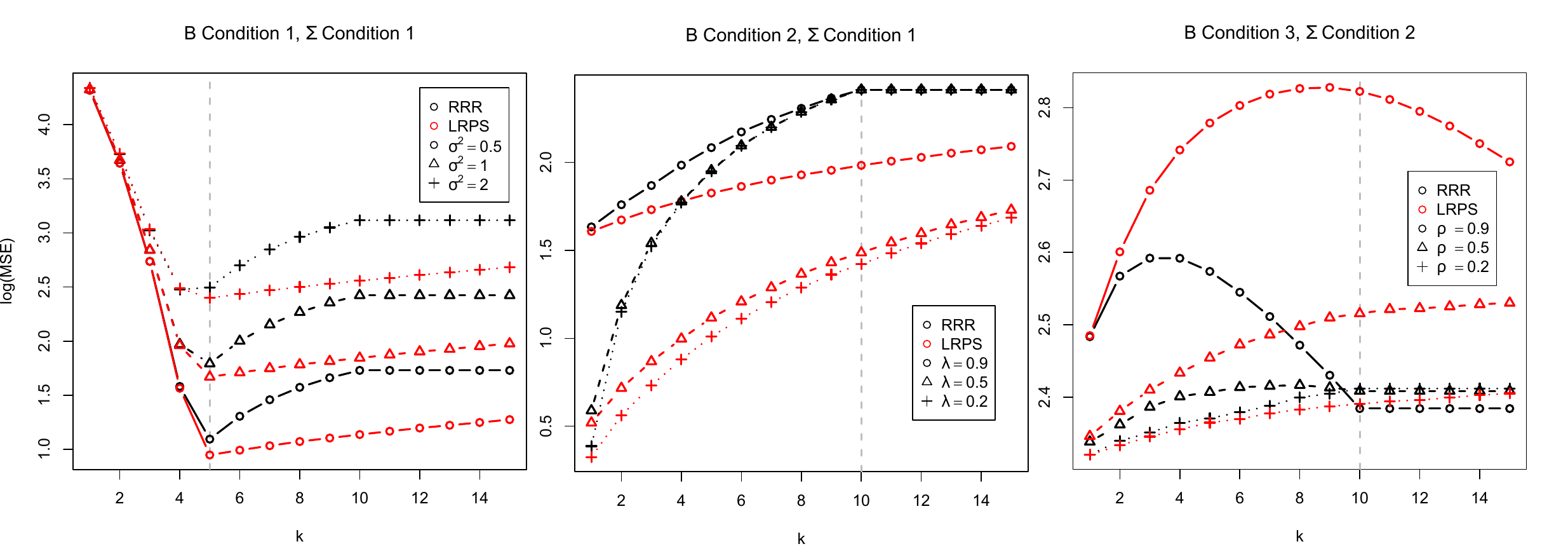}

\caption{Comparison of LRPS and RRR performance as a function of $k$ in the
case where $p=10$, $q=100$, $n=100$. Black lines denote the performance
of RRR and red those of LRPS. In the middle and right plots, we fix
$\sigma^{2}=1$ and vary the eigenvalues of the signal and noise respectively
by adjusting $\lambda$ and $\rho$. In $B$ Condition 3, we set all
the singular values of $B$ to be one, and maintain that $r(B)=p$
indicated by the dashed line, in the left plot, we set the sparsity
to be $s=5$.\label{fig:res_large_q}}
\end{figure}

\subsubsection*{Moderate $q$, large $p$}

In contrast to the high $q$ setting, in this simulation we investigate the relatively
high-dimensional estimation problem where $q=10$ is fixed and we
set $p\in\{20,40,60\}$, we fix the sample size at $n=100$. In this
situation, the true rank of $B$ is limited to be less than or equal to $q=10$. Specifically, we simulate $B$ matrices according to Condition 3, and set
the rank $k_{*}=5$. We vary the SNR according by adopting $\Sigma$-Condition
1, and letting $\sigma^{2}\in\{0.5,1,2\}$. The results of the experiment
are seen in Figure \ref{fig:fixed_q_growing_p}. We see that both
RRR and LRPS perform well relative to the OLS, and whilst RRR outperforms LRPS
for $k\le k_{*}$ the difference between the two estimators is small.
In this case, the best approximating model, i.e. the $k$ which minimises
the MSE, depends not only on the SNR mediated by $\sigma^{2}$, but
also the variance in the estimators due to the higher number of covariates
$p>q$. As $p$ approaches $n$ from below, we see that again the
rank-1 approximation achieves the best performance for both LRPS and RRR.
These results highlight the importance of performing the pre-smoothing
(projection) to reduce noise, either when the dimensionality of covariates
is high, and/or when the SNR is low. Generally when $k\ge\hat{k}=\arg\min_{k}\mathrm{MSE}(k)$,
we see that LRPS outperforms RRR, and the difference in performance
can be substantial in the large $q$ setting.

\begin{figure}
\includegraphics[width=1\columnwidth]{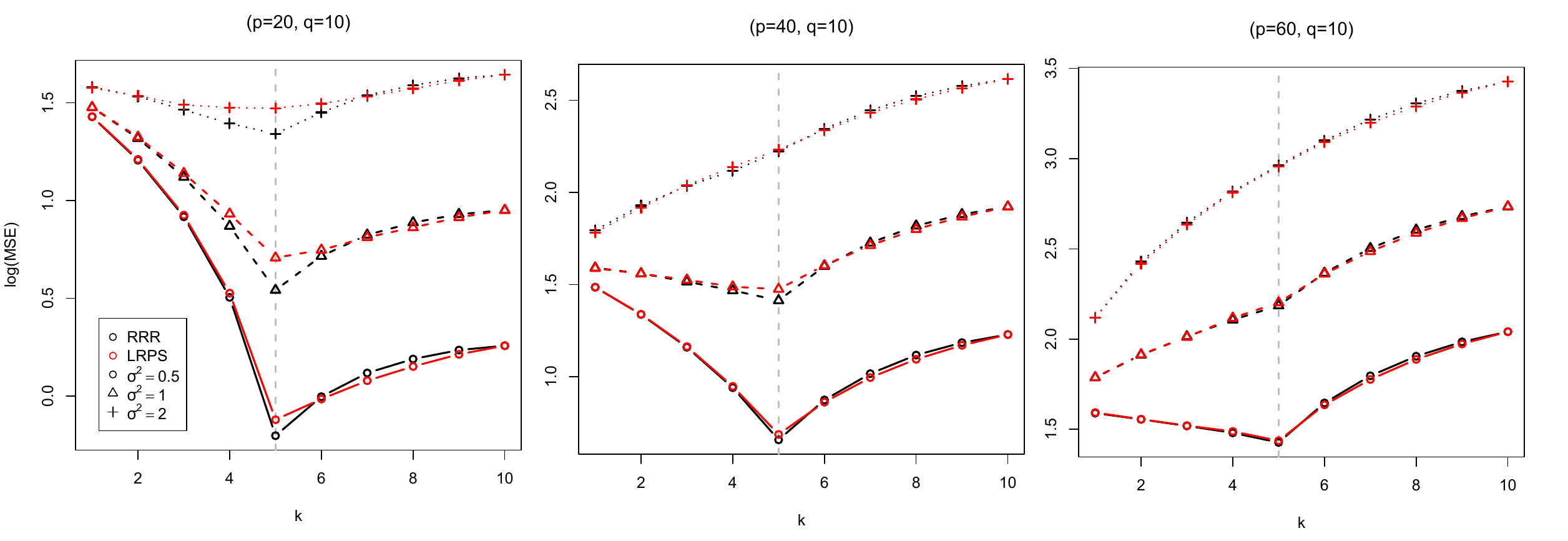}

\caption{Comparison of LRPS and RRR performance as a function of $k$ in the
case where $q=10$, $n=100$ and $p\in\{20,40,60\}$. The covariance
of errors is given by $\Sigma$-Condition 1 with $\sigma^{2}\in\{0.5,1,2\}$,
the true rank of $B$ is denoted by the vertical grey line.\label{fig:fixed_q_growing_p}}
\end{figure}

\paragraph{Sensitivity to heavy-tailed errors.}
To conclude our investigations into empirical estimation error, we also examine the performance of LRPS compared with RRR under heavy-tailed error settings. Specifically,  multivariate heavy tail errors are generated from a multivariate $t$-distribution with  3 degrees of freedom (using the \texttt{mvtnorm} \textit{R} package \cite{mvtnorm2009}). The simulation setup follows the large $q = 100$, moderate $p = 10$ and $n = 100$ design while varying the structure of the coefficient matrix $B$. The error covariance is taken to be either diagonal or Toeplitz, matching the two different noise conditions for the covariance structures used in the Gaussian setting. The results are shown in Figure \ref{fig:heavytail}. 

In the left panel, corresponding to a low-rank sparse coefficient matrix $B$, RRR performs well when the rank is correctly specified but suffers a notable loss in accuracy when the rank is overestimated. By contrast, in the two right panels, where $B$ is dense, LRPS consistently performs better than RRR across the range of ranks considered. Overall, the performance gains of LRPS relative to OLS and RRR appear robust to the case of heavy-tailed errors.

\begin{figure}[!h]
\includegraphics[width=1\columnwidth]{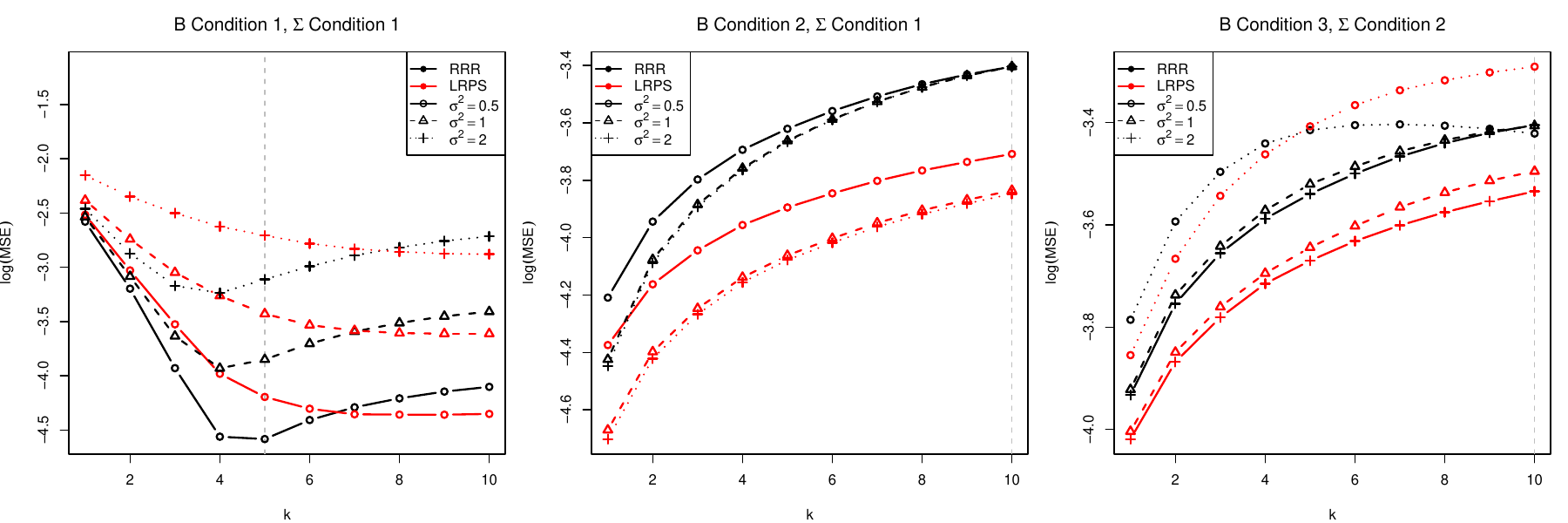}
\caption{Comparison of coefficient estimation error across LRPS and RRR under a heavy-tailed error. All results are averaged over 100 simulations with $p = 10$, $q = 100$, and $n = 100$ with different $B$ conditions.}
\label{fig:heavytail}
\end{figure}

\subsection{Predictive performance}
For LRPS and RRR, we see that the choice of $k$ is
important. In this section, we demonstrate the effectiveness of cross-validation for choosing $k$, and further evaluate the predictive performance of the LRPS and RRR estimators, in an out-of-sample setting.

Before considering methods to select $k$, we first consider the sensitivity of the estimation performance to mis-specification of $k$ in some examples. To this end, Figure \ref{fig:simsensitivity} presents plots of the mean squared predictive error as a function of the rank parameter $k$ for LRPS and RRR, under the large $q$, moderate $p$ setting corresponding to Figure \ref{fig:res_large_q}. The performance trends observed in the predictive plots are consistent with those seen in the estimation of the coefficient matrix, as shown in Figure \ref{fig:res_large_q}. Notably, in scenarios where the error spectrum is relatively flat and LRPS is expected to perform well, it outperforms RRR (the left two panels). In contrast, when the error covariance $\Sigma_e$ in settings where the error covariance exhibits exhibits a spiked spectrum (the right panel), LRPS performs less favourably, highlighting the limitations of the method under such conditions. 

\begin{figure}[!h]
\includegraphics[width=\columnwidth]{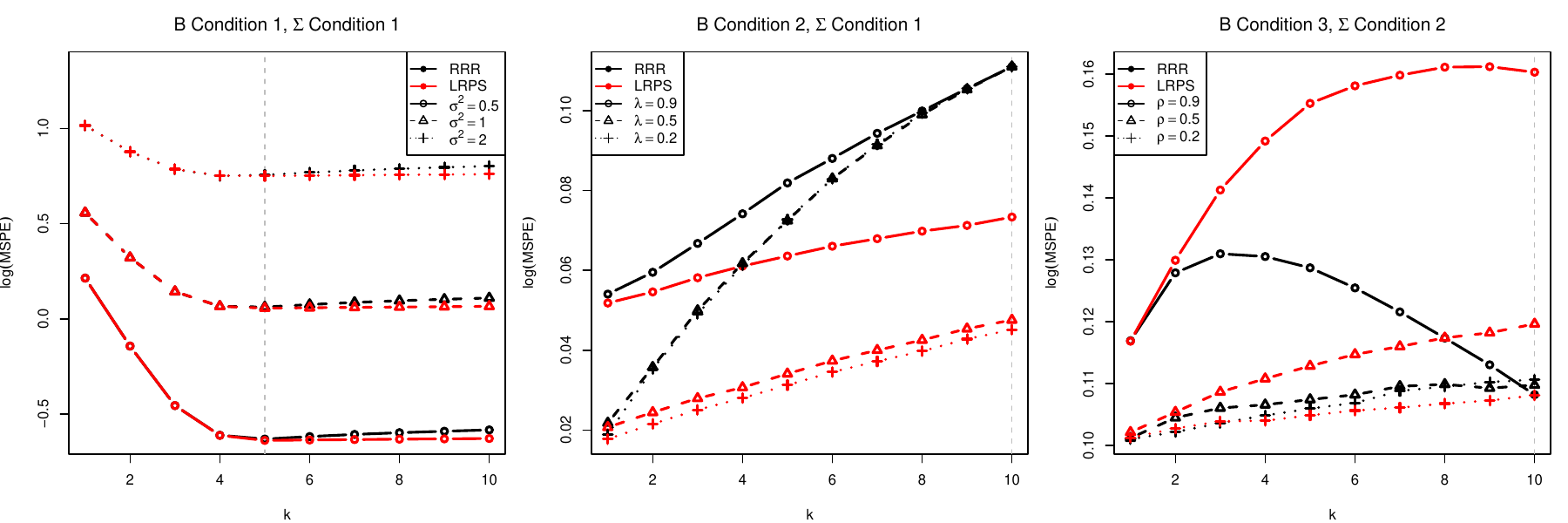}
\caption{Sensitivity plots of out-of-sample prediction risk versus $k$ for LRPS and RRR in the case where $p=10$, $q=100$, $n=100$. Black lines denote the performance of RRR, and red denotes the performance of LRPS.}
\label{fig:simsensitivity}
\end{figure}

In general, the selection of how strong we wish our constraints to be in statistical inference may depend on the task at hand. For example, in the lasso ($\ell_1$ regularised regression), if one wishes to perform an optimal selection of non-zero coefficients, the specification of the regularisation (constraint) may be different from that for producing optimal predictions. In that setting, typically cross-validation will overestimate the support of the regression function compared to more conservative information criteria. Moving back to the rank-restricted setting with RRR and LRPS, we may encounter a similar position---that is, do we construct an information criterion for the selection of $k$, or do we use some in-sample empirical measure of prediction error, e.g. cross-validation. As we are focusing on global constraints (i.e. working with linear projections across all variates), we here focus on predictive performance rather than parametric inference when deciding the appropriate specification of $k$, and to this end, adopt a simple cross-validation procedure.

Specifically, we propose to apply 2-fold cross-validation, such that we minimise
\begin{equation}
\mathrm{cvMSPE}(k):=\frac{1}{2}\left(\|Y^{(1)}-X^{(1)}\hat{B}^{(2)}W_k^{(2)}\|_F^2 +
\|Y^{(2)}-X^{(2)}\hat{B}^{(1)}W_k^{(1)}\|_F^2\right)\;,\label{eq:cvMSPE}
\end{equation}
where $Y^{(1)}$ indicates data obtained from the first half of the data and $Y^{(2)}$ the second half of the data. Note, we use the eigenvectors from the corresponding training dataset, i.e., $W_k^{(j)}=V_k^{(j)}V_k^{(j)\top}$ are based on the first $k$ singular vectors of $Y^{(j)}\in\mathbb{R}^{n/2\times q}$, where $j \in \{1,2\}$.  The $\mathrm{cvMSPE}(k)$ represents a popular in-sample measure of the out-of-sample performance, in particular giving (for large $n$) a good approximation to a form of predictive risk \cite{bates2024cross}. 

In the following experiments, we consider the estimates $\hat{k}$ according to 
\[
\hat{k} = \arg\min_{k=1,2,\ldots} \mathrm{cvMSPE}(k)\;,
\]
for both RRR and LRPS. Note, the version stated in (\ref{eq:cvMSPE}) can be readily applied to RRR by simply replacing $V_k$ with $U_k$ (from Definition \ref{dfn:rrr}). To evaluate the performance of the selection criteria, we consider the empirical distribution of $\hat{k}$ alongside the out-of-sample expected prediction error $\mathbb{E}_{(y',x')}\|y'-x'\hat{B}V_{\hat{k}}V_{\hat{k}}^\top\|_F^2$
where $y',x'$ are assumed to be generated from the same generating process that generated the training data, i.e. $y'\sim\mathcal{N}_q(0,B^\top \Sigma_X B + \Sigma_e)$ and $x'\sim\mathcal{N}_p(0,I_p)$.

\begin{figure}[!h]
\includegraphics[width=1\columnwidth]{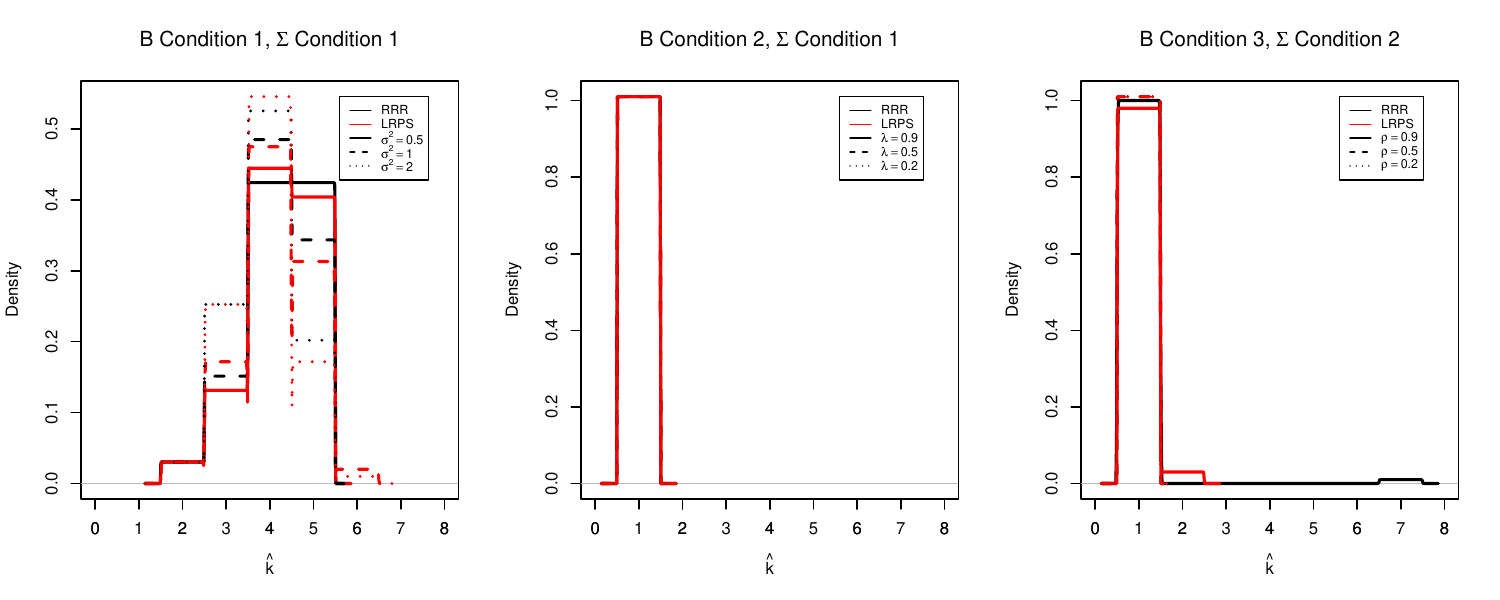}
\includegraphics[width=1\columnwidth]{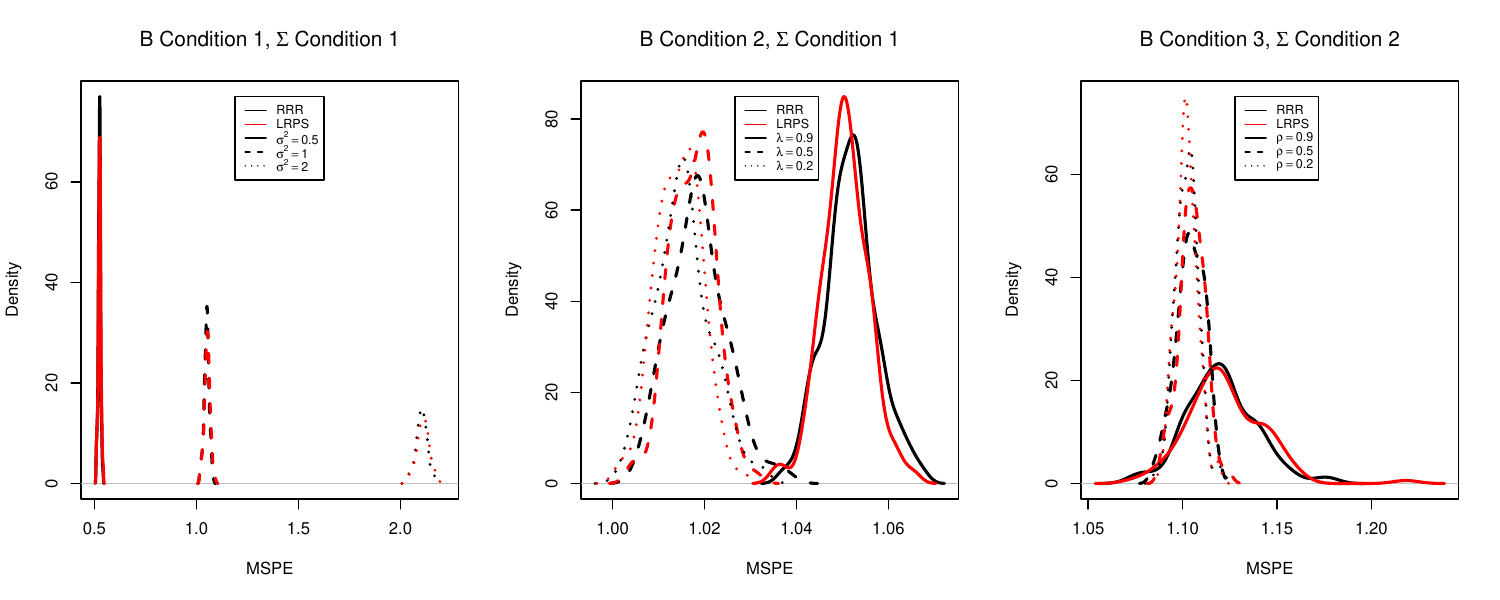}

\caption{Top: Distribution of $\hat{k}$ (chosen by 2-fold cross-validation); Bottom: Distribution of prediction error. Both distributions are estimated across $100$ replications. Prediction error is based on applying the singular vectors from the training set $n=100$ to the out-of-sample data. The simulation settings mirror those in Figures \ref{fig:singularvals} and \ref{fig:res_large_q}, with RRR given by the black lines, and LRPS given by those in red.  \label{fig:pred_err}}
\end{figure}

The results of our experiment are presented in Figure \ref{fig:pred_err}. The simulation settings are exactly the same as those considered in Figure \ref{fig:res_large_q}. In this case, the difference between LRPS and RRR is negligible---both LRPS and RRR estimators appear to adapt its rank (and model complexity) via the proposed cross-validation procedure, i.e. $\hat{k}$ depends on the simulation design (e.g. signal-to-noise ratio, noise/signal distribution).

It is worth noting that, in this study, we adopt a 2-fold cross-validation scheme for simplicity; however, this is not the only approach for selecting the rank parameter $k$. In other applications, $k$ may be chosen adaptively for each task or fixed across different experimental settings. For example, in time-series neuroscience applications, the optimal rank can be selected using rolling-window or blocked cross-validation schemes that respect temporal dependence. In multi-subject studies, the rank parameter may also be tuned at the subject level or at the group level, depending on whether the primary interest lies in individual-specific or population-level dynamics.



\subsection{Performance comparison with PCR}\label{sec:compwithpcr}
When $p$ is large, it is natural to consider PCR, as it captures the most relevant structure in $XB$. While a comprehensive comparison among LRPS, RRR, and PCR is of interest and will be pursued in future work, we focus here solely on a comparison between LRPS and PCR. Using the same experimental settings as in Figure \ref{fig:res_large_q}, we fix $q = 100$, $n = 100$, and $p = 10$, and vary the conditions of $B$ and $\Sigma_e$ through different values of $\lambda$ and $\rho$. Under these settings, we examine and compare the coefficient estimation performance of LRPS and PCR.

Figure \ref{fig:lrps_pcr} demonstrates the performance of LRPS relative to PCR across various settings. The results exhibit patterns similar to those observed in the comparison between LRPS and RRR. In scenarios where LRPS outperforms RRR, it also maintains an advantage over PCR. Conversely, in settings where LRPS is expected to perform worse than RRR, its performance likewise falls short when compared with PCR. Overall, these findings reinforce the same conclusion, such that in cases with a flat error spectrum, LRPS continues to achieve the best performance.

\begin{figure}[!h]
\includegraphics[width=\columnwidth]{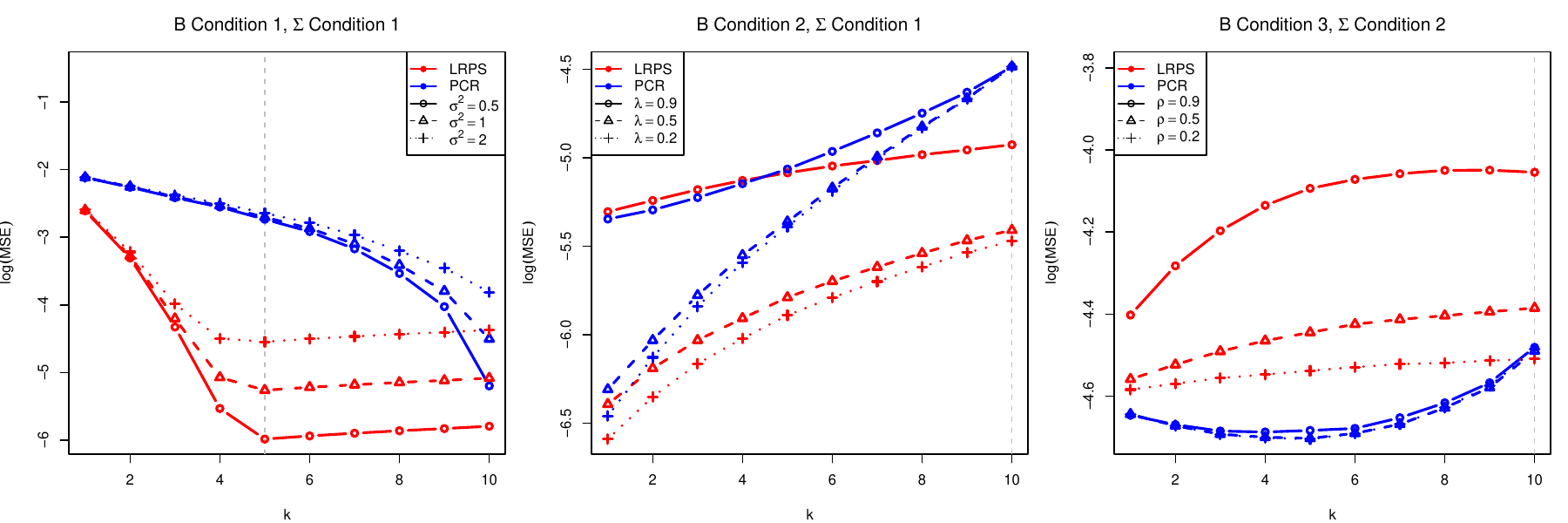}
\caption{Comparison of LRPS and PCR performance as a function of $k$ in the
case where $p=10$, $q=100$, $n=100$. Red lines denote LRPS and blue lines denote PCR.}
\label{fig:lrps_pcr}
\end{figure}

\subsection{Computational cost investigation}\label{sec:compcost}

To evaluate and compare the computational efficiency of the three methods OLS, RRR, and LRPS, we conduct a series of simulations by systematically varying one of the dimensions using scaling factors from 1 to 5 (sample size $n$, response dimension $q$ or covariate dimension $p$) while keeping the others fixed. For simplicity, we fix $k = 1$ for both LRPS and RRR. 

\subsubsection*{The case when $p > q$}
In the $p > q$ case, we fix $q = 10$ and let either $n$ or $q$ increase with the scaling factor, starting from $n = 200$ and $p = 20$. For each setting, we generated simulation datasets and measured the average computation time over 100 repetitions for each method. The top plots in Figure \ref{fig:scaling1} display the average computation time relative to OLS, where the computational time of OLS is normalised to 1 at each scaling level. In other words, the plotted values for RRR and LRPS reflect their runtime as a multiple of the corresponding OLS runtime. This illustrates the relative computational cost of adopting one of the low-rank regression methods as data dimensionality increases. It is evident that as $n$ or $p$ increase, LRPS consistently outperforms RRR in terms of computational time relative to OLS (Figure \ref{fig:scaling1}, top-left and top-right panels). When examining the within method computation time relative to the runtime for the initial simulation setup, we observe that there is a mildly worse scaling with increasing $n$ (Figure \ref{fig:scaling1}, bottom-left). computational cost of LRPS increases at a faster rate than that of RRR. However, when examining the scaling of the methods with increasing $p$, whlist all methods have a quadratic scaling behaviour, LRPS increases at a slower rate than both OLS and RRR. 

\begin{figure}[!h]
\includegraphics[width=1\columnwidth]{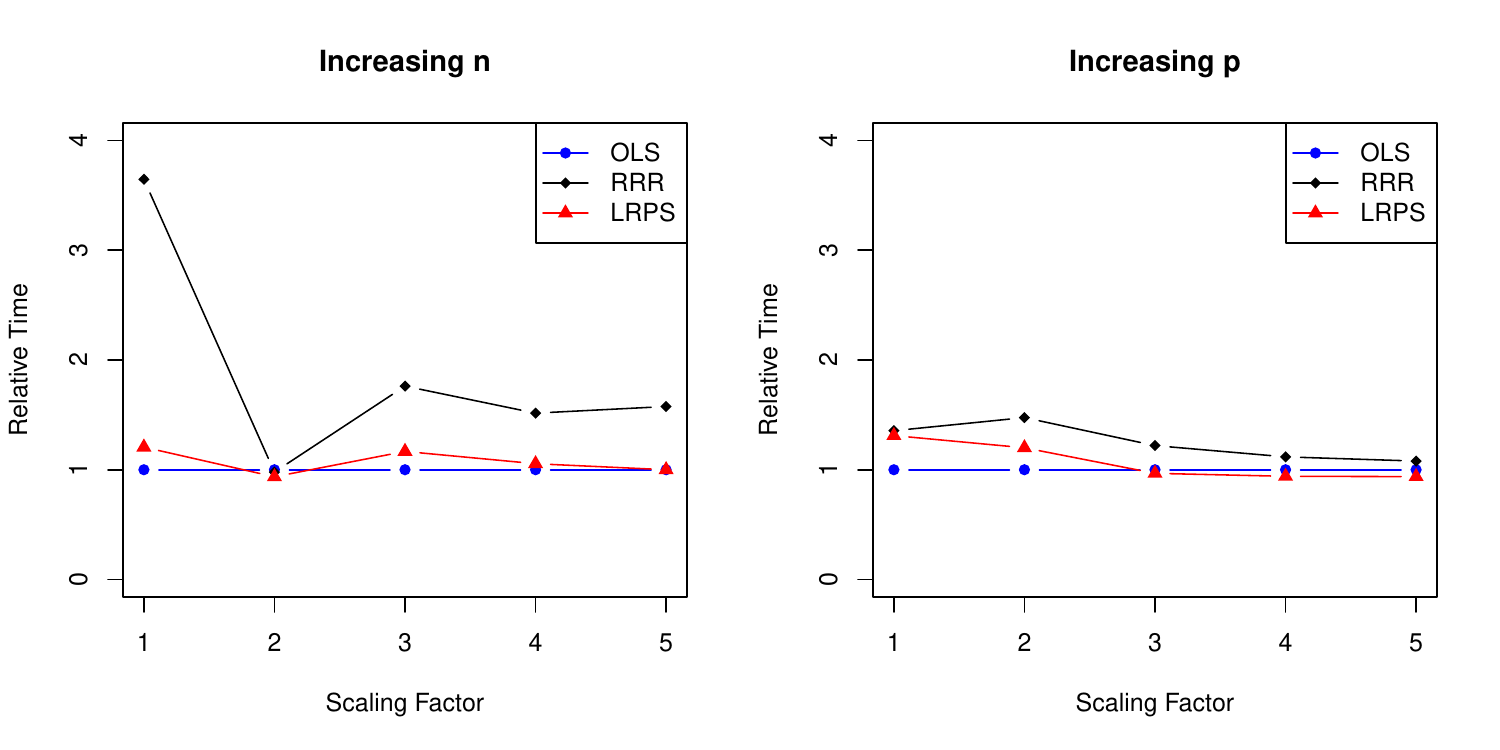}
\includegraphics[width=1\columnwidth]{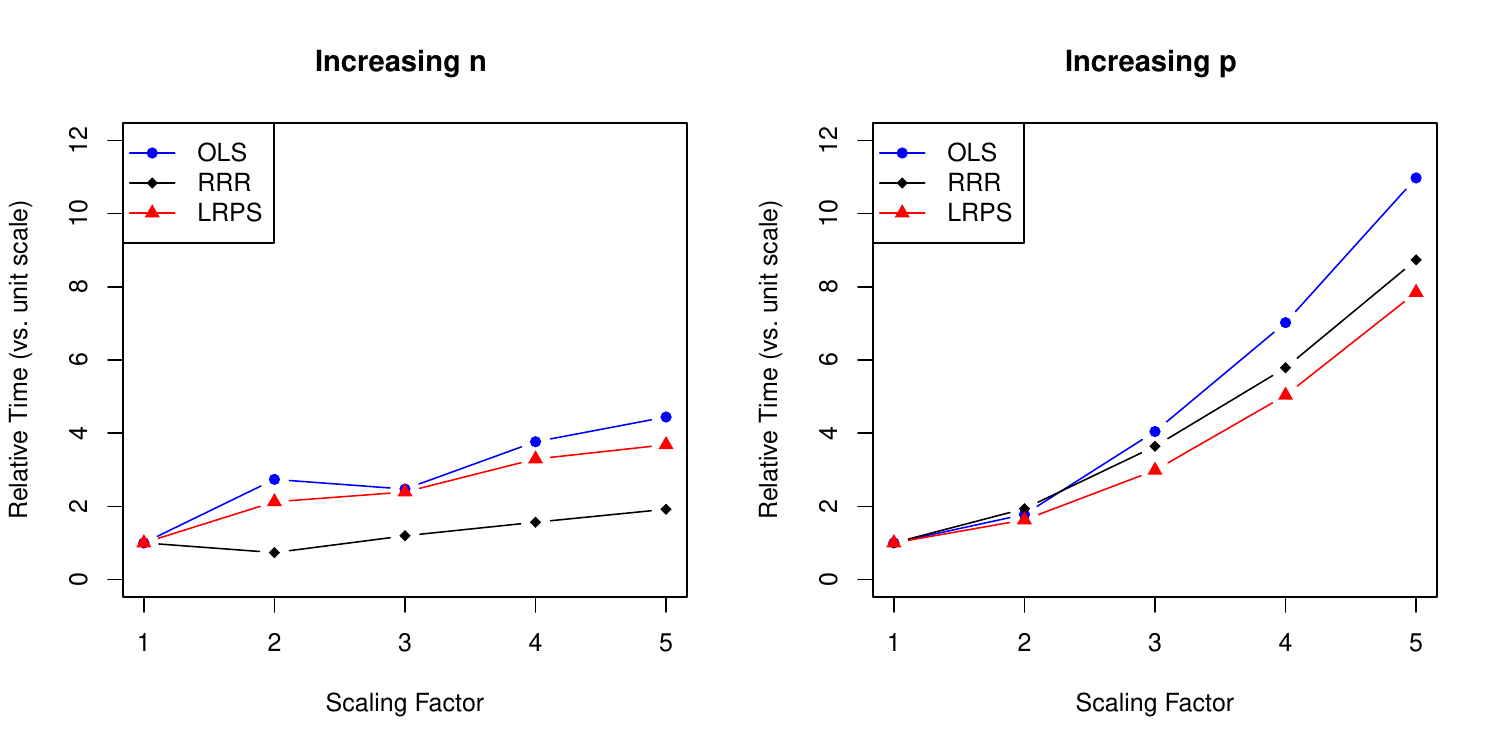}
\caption{Top panels: Computation time of RRR and LRPS relative to OLS across varying data dimensions, where OLS time is normalised to 1 at each scaling level. Bottom panels: Computation time of OLS, RRR and LRPS relative to their own runtime at scaling factor 1 under varying data dimensions. Left: Varying $n$ with $q = 10$ and $p = 20$; Right: Varying $p$ with $n = 200$ and $q = 10$. }
\label{fig:scaling1}
\end{figure}

\subsubsection*{The case when $q > p$}

In the $q > p$ case, we always fix $p = 30$ and let either $n$ or $q$ increase with the scaling factor, starting from $n = 60$ and $q = 100$. Similar to above, for each setting and each method, we calculate the average computation time over 100 simulated datasets. The top plots in Figure \ref{fig:scaling2} display the average computation time relative to OLS. Notably, as $n$ or $q$ increases, LRPS consistently outperforms RRR in terms of computational efficiency, demonstrating faster runtimes relative to OLS. The bottom plots show the computational cost of each method relative to its own runtime at the scaling factor of 1. From the plots, we observe that the growth in computation time for LRPS is consistently lower than that of RRR, indicating better scalability with respect to increasing data dimension.  In particular, both the OLS and LRPS illustrate linear growth with $q$, whereas RRR has a mild quadratic scaling. 
\begin{figure}[!h]
\includegraphics[width=1\columnwidth]{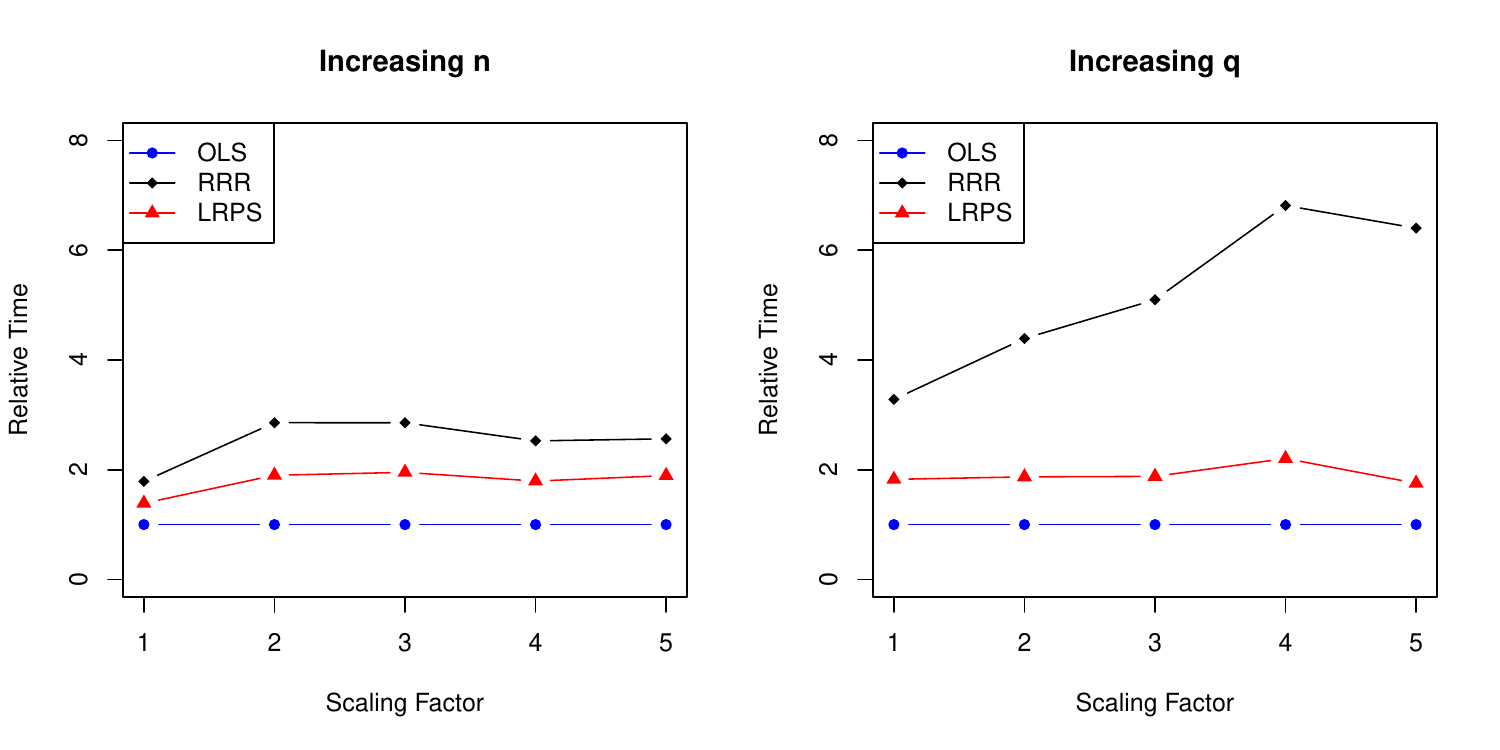}
\includegraphics[width=1\columnwidth]{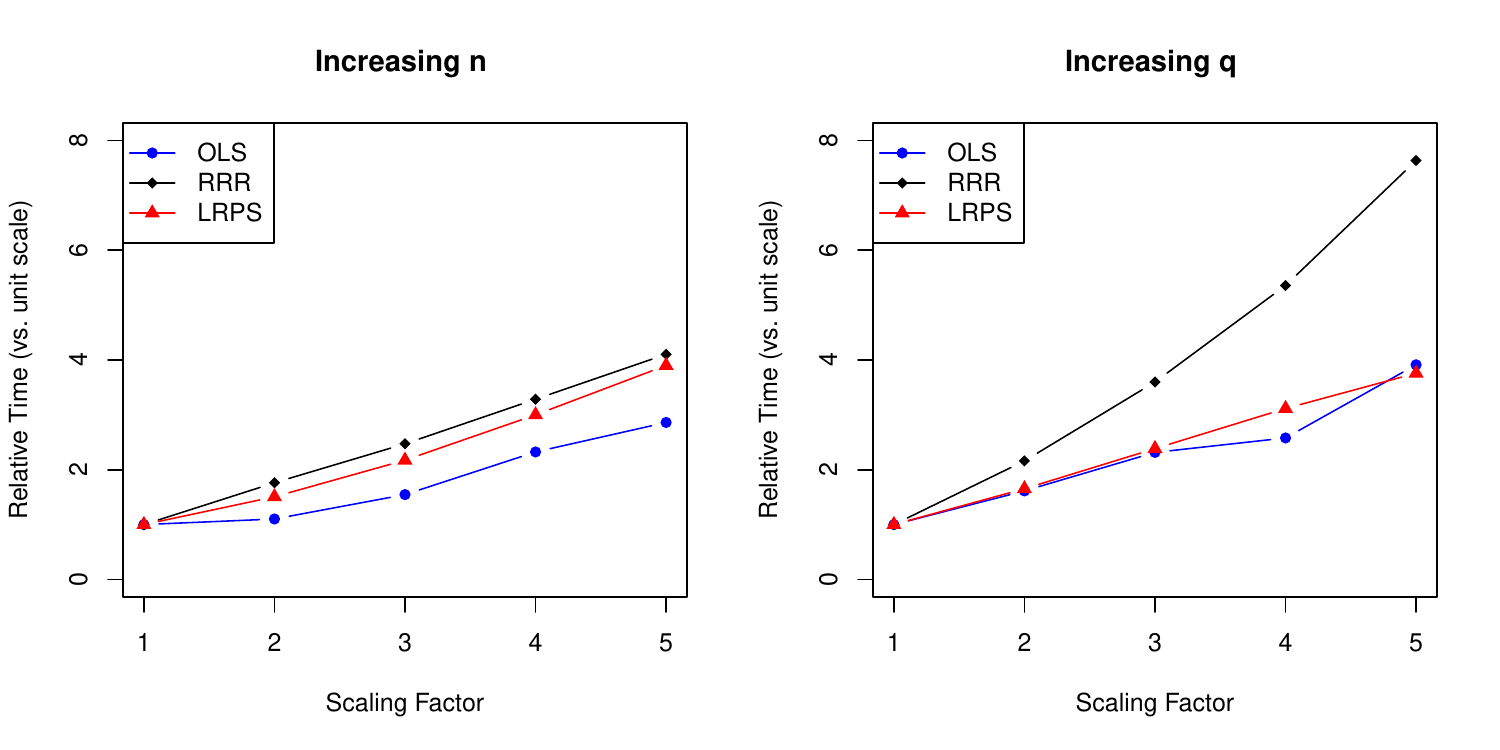}
\caption{Top panels: Computation time of RRR and LRPS relative to OLS across varying data dimensions. OLS time is normalised to 1 at each scaling level. Bottom panels: Computation time of OLS, RRR and LRPS relative to their own runtime at scaling factor 1 under varying data dimensions. Left: Varying $n$ with $p = 30$ and $q = 100$; Right: Varying $q$ with $n = 60$ and $p = 30$. }
\label{fig:scaling2}
\end{figure}
These simulations confirm the observations in Section \ref{sec:complexity}, where RRR was shown to have additional computational complexity compared to LRPS.

We also recorded raw empirical runtime and memory usage of LRPS when compared to other multiresponse regression methods.  More specifically, we evaluate the OLS, RRR, LRPS and PCR techniques investigated in this section in a range of regression scenarios.  The raw wall clock computation times are shown in Table \ref{tab:runtimep} in Appendix \ref{app:runtime} for when q is fixed, examining performance as $n$ or $p$ increases, whilst Table \ref{tab:runtimeq} fixes $p$ and examines performance as $n$ or $q$ increases. For the first setting, LRPS shows a slower increase in runtime compared to other methods as the sample size increases; as $p$ increases all methods are comparable.  When $p$ is fixed, PCR appears the most stable as $n$ increases, whereas when $q$ increases LRPS and PCR are comparable.  Tables \ref{tab:memoryp} and \ref{tab:memoryq} report the corresponding memory consumption for experiments examined in Tables \ref{tab:runtimep} and \ref{tab:runtimeq}; the simulations indicate that all methods show stable memory usage across different regression problems.


\section{Application to real data}

\label{sec:data} In this section, we demonstrate our proposed LRPS method to study datasets arising from two scientific applications, the first considering air pollution and the second gene activation data.

\subsection{Air pollution example}

Particulate matter up to $2.5\mu m$ (PM2.5), is a type of air pollution
that poses a significant threat to human health and the environment.
These tiny particles are made up of a complex mixture of solids, chemicals,
and liquid droplets that can penetrate the respiratory system and
cause respiratory and cardiovascular problems. It is recognised that controlling such particular matters is of vital importance for the public health \citep{franck2011effect,pun2017long}. In this example, we consider the possible association between PM2.5 emissions and recordings of other pollutant gases, namely: ozone (O3), sulphur dioxide (SO2), carbon monoxide (CO), and nitrogen dioxide (NO2). To study this association,
we consider three datasets:
\begin{itemize}
\item \textbf{UK:} Pollution across seven cities ($p=7)$ in the United
Kingdom (UK). There are daily recordings of the pollutants for a total
of $n=400$ data-points, $q=28$ (the number of sites $\times$ number
of polluting gases). This data is available from \url{https://uk-air.defra.gov.uk/}. We
use the first  80\% observations $n_{\mathrm{train}}=320$ data points to train the model,
and the remaining 20\% observations $n_{\mathrm{test}}=80$ for evaluation.
\item \textbf{Beijing:} This multi-site air-quality dataset \citep{liang2015assessing}
can be found at the UCI Machine Learning Repository {\url{http://archive.ics.uci.edu/ml/datasets}}
and has been used in the literature in different contexts, e.g. \citet{ma2022multiple}.
In this case, we have hourly air pollutant measurements from $p=12$
outdoor monitoring sites in Beijing, collected from March 2013 to
February 2017. We focus on the air pollutant measurements for January
2017, resulting in a dataset with $n=743$ ($n_{\mathrm{train}}=594$, 
$n_{\mathrm{test}}=149$) observations and $q=48$ outcomes. It should
be noted that this dataset contains some missing values, which for simplicity we replace by the value of the previous observation in the missing variable. 
\item \textbf{USA:} The last dataset is collected from $p=37$ monitoring sites in the United States
of America (USA) from January 2017 to April 2019, resulting in a total
of $n=728$ ($n_{\mathrm{train}}=578$, $n_{\mathrm{test}}=150$)
observations (available from \url{https://www.epa.gov/outdoor-air-quality-data}). This is the largest dataset in this study, with $q=148$
outcomes.  
\end{itemize}

For each dataset we take first-differences (across rows of $Y,X$)
and standardise all variables so they have mean zero and standard deviation
1. This preprocessing follows established practice in previous studies \cite{ma2022multiple, zou2022estimation} and is adopted to promote stationarity of the data. In each case, we estimate the regression parameters $B$ on the
training dataset, and tune $k$ via the 2-fold cross-validation procedure
employed previously. The remaining $n_{\mathrm{test}}$ data points (20\%)
are used for evaluation via the mean-square prediction error

\[
\text{ MSPE }=\frac{1}{qn_{test}}\left\Vert X_{\text{test }}\hat{B}-Y_{\text{test }}\right\Vert _{F}^{2}.
\]

The results of the prediction errors and the rank $k$ used in the
low-rank pre-smoothing model are presented in Table \ref{app:result1}.

\begin{table}[!h]
\centering \caption{Prediction error and optimal rank for the pollution datasets described in the main text. \label{app:result1}}
\begin{tabular}{lccc}
\toprule 
\textbf{Dataset} & \textbf{Method} & \textbf{Prediction Error (MSPE)} & \textbf{Rank} \\
\midrule 
\multirow{3}{*}{UK ($p=7,q=28$)} & LRPS  & \textbf{0.738}  & 5 \\
 & RRR  & 0.749  & 6 \\
 & OLS  & 0.751  & N/A \\
\midrule
\multirow{3}{*}{Beijing ($p=12,q=48$)} & LRPS  & \textbf{1.145}  & 12 \\
 & RRR  & 1.154  & 6 \\
 & OLS  & 1.175  & N/A \\
\midrule
\multirow{3}{*}{USA ($p=37,q=148$)} & LRPS  & \textbf{0.911}  & 2 \\
 & RRR  & 0.916  & 1 \\
 & OLS  & 0.925  & N/A \\
\bottomrule
\end{tabular}
\end{table}

In all cases, LRPS provides better predictive performance than the
more traditional approaches. Generally, LRPS and RRR outperform OLS. Whilst the differences in performance are not
vast, the results demonstrate that low-rank approximations may
generally be useful, and it is interesting that in these cases the
rank selected is always lower than $\min(p,q)$. In the case of the
USA, we see that the rank selected is $k=1,2$ in RRR and LRPS respectively.
This behaviour was seen in our synthetic experiments in the high $q$
setting, e.g. see $B$ Condition 2, $\Sigma$ Condition 1 results
in Figure \ref{fig:res_large_q}.

\subsection{Gene association example}

In addition to the air pollution datasets, we applied our proposed
methodology to a genetic association dataset from \citet{wille2004sparse}.
This dataset is from a microarray experiment designed to explore regulatory
mechanisms within the isoprenoid gene network of \emph{Arabidopsis
thaliana}, also known as \emph{thale cress} or \emph{mouse-ear cress}.
Isoprenoids are known to play numerous essential biochemical roles
in plants. To track gene expression levels, $n=118$ GeneChip microarray
experiments were conducted. The dataset includes $p=39$ predictor
genes from two isoprenoid biosynthesis pathways, MVA and MEP, while
the response variables consist of the expression of $q=795$ genes
across 56 metabolic pathways. To reduce skewness in the distributions, all variables are log-transformed before fitting the linear models. We use the first $n_{\mathrm{train}}=96$ observations as the training set, and then calculate the prediction error using the remaining $n_{\mathrm{test}}=22$ observations.

\begin{table}[!h]
\centering \caption{Prediction error and optimal rank for the genetic association dataset described in the main text. \label{app:result4}}
\begin{tabular}{lcc}
\toprule 
\textbf{Method}  & \textbf{Prediction Error (MSPE)}  & \textbf{Rank} \\
\midrule 
LRPS  & \textbf{0.414}  & 2 \\
RRR  & 0.418  & 2 \\
OLS  & 0.481  & N/A \\
\bottomrule
\end{tabular}
\end{table}

The results are presented in Table \ref{app:result4}, and demonstrate
the benefits of the low-rank approximation similar to the air-pollution
dataset. In this case, we have a very large $q$ compared to $p$,
and both are relatively large when compared to the number of data-points
$n$. It is worth remarking that our holdout dataset in this case
is only 
 $n_{\mathrm{test}}=22$ 
observations, however, the performance
is also averaged across the number of outcomes $q=795$ which is relatively
large. The analysis here can be considered relatively ``high-dimensional''
in nature. Both LRPS and RRR choose $k=2$ as the best approximating
rank, potentially illustrating behaviour akin to $B$ Condition 1/2
and $\Sigma$ Condition 1 in our synthetic experiments, i.e. where
the error spectrum was relatively flat, but the signal was concentrated
in terms of the top singular values. Whilst the difference between
LRPS and RRR is not large, LRPS maintains the better performance (in-line
with the synthetic experiments and the other application), and both
low-rank methods considerably outperform OLS.

\subsection{Sensitivity to rank $k$}

Figure \ref{fig:appsensitivity} presents a sensitivity analysis of the rank parameter $k$ with respect to predictive performance for both the LRPS and RRR methods. In datasets with a small number of observations, such as the UK dataset, the optimal performance may occur at very low ranks (e.g., $k = 1$), yet cross-validation can tend to overestimate the rank in this regime. In contrast, for larger datasets, cross-validation is no longer constrained to selecting a rank of 1 and instead is able to identify higher-rank structures that improve performance for the application.

\begin{figure}[!h]
\includegraphics[width=0.5\columnwidth]{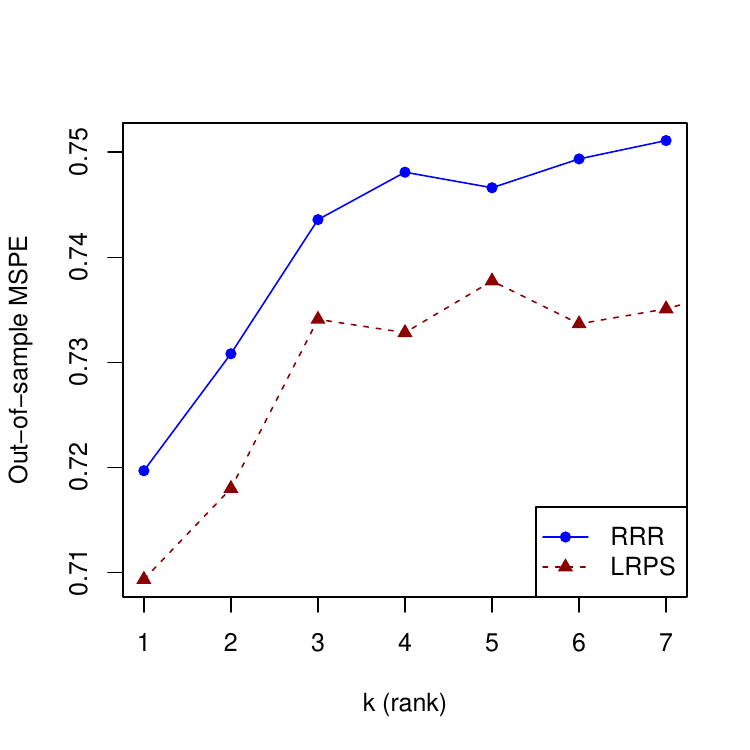}
\includegraphics[width=0.5\columnwidth]{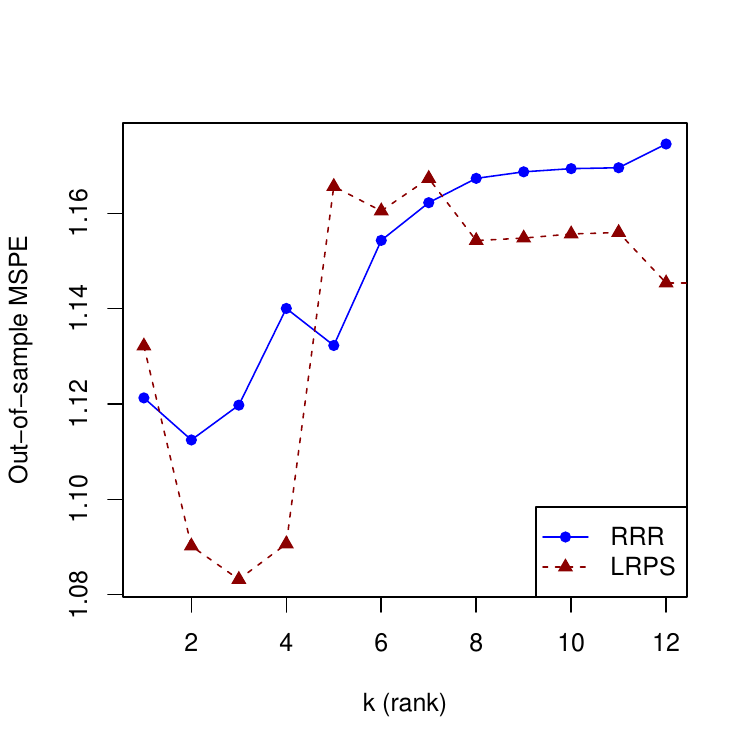}
\includegraphics[width=0.5\columnwidth]{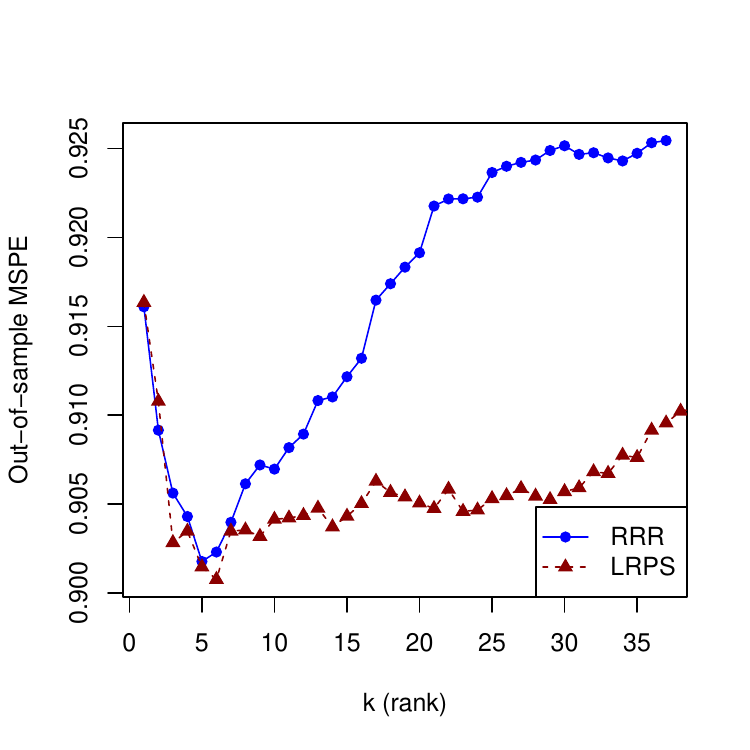}
\includegraphics[width=0.5\columnwidth]{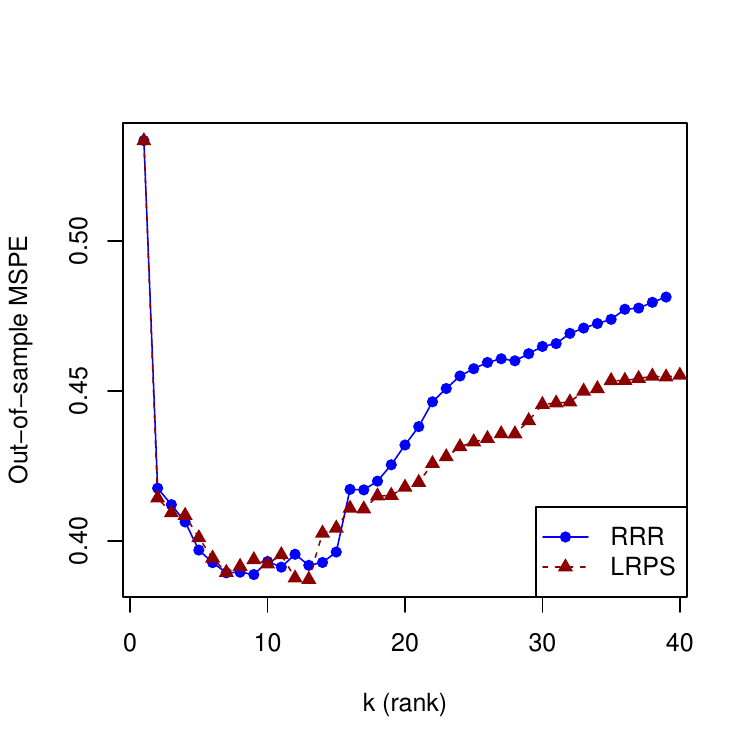}
\caption{Sensitivity plots of out-of-sample prediction risk versus $k$ for LRPS and RRR on real datasets. Top-left: UK dataset; Top-right: Beijing dataset; Bottom-left: US dataset; Bottom-right: Genetic dataset.}
\end{figure}
\label{fig:appsensitivity}

\section{Conclusions}\label{sec:concs}

This article proposes {\em Low-Rank Pre-Smoothing}, a technique for multi-response regression settings.  When viewed as a projection, the method can be seen as an initial step to smooth observed data to increase its signal-to-noise ratio, after which a traditional estimation method such as ordinary least squares (OLS) can be performed for estimation and subsequent prediction.  Our technique outperforms the traditional OLS estimator in terms of mean square estimation and out-of-sample prediction error, driven by sacrificing a small amount of bias in favour of reducing the estimation uncertainty. This improved performance was demonstrated through extensive simulated examples, our methodology showing particular benefits when a large number of responses is observed and/or when the signal-to-noise ratio is low. We also illustrated the utility of our proposed methodology on a number of real datasets. We believe our proposed LRPS estimator provides a simple low-rank regression method that nicely complements existing methods, and can provide a further multi-response regression alternative to a wide range of practitioners. \\
Reflecting on the limitations of the technique, we do note that in some circumstances it may indeed be preferable to use an alternative low-rank regression technique, such as reduced rank regression (RRR) or principal component regression (PCR), and in general the best method will depend on the true generative process. In particular, one may consider using an alternative methodology when the signal $XB$ aligns predominantly in a direction that is not represented by $V_k$, as in these cases, whilst the variance may be significantly reduced by LRPS, there is also significant bias. On the other hand, when $q$ is modest, variance savings may be limited by LRPS, and RRR or PCR may be more beneficial as they focus more closely on the space spanned by the design matrix. In practice, one can choose between the alternative methods by assessing out-of-sample prediction error via cross-validation procedures, which could in principle be extended to choose between low-rank regression methods, alongside an appropriate choice of $k$.

Whilst the motivation for the work in this article arose from parameter estimation and prediction in the classical multi-response setting, we envisage that the tools developed here have the potential to make similar gains in other settings, such as model selection in multi-response regression, nonlinear regression models, and time series forecasting. In particular, one non-trivial extension to the LRPS technique would be in the high-dimensional data setting, which would require development of penalised (LASSO and ridge) regression techniques using LRPS.

\section*{Author Contributions}
All authors have contributed equally to this work.

\section*{Acknowledgements}

This research did not receive any specific funding from agencies in the public, commercial, or not-for-profit sectors.

\section*{Conflict of Interest}
The authors declare no potential conflicts of interest.

\section*{Data Availability Statement}
The datasets used to support the work in this article are available from the links provided in the text of Section \ref{sec:data}.



\bibliographystyle{unsrtnat}
\bibliography{references}

@book{abadir,
  title={Matrix Algebra},
  author={Abadir, K. M. and Magnus, J. R.},
  volume={1},
  year={2005},
  publisher={Cambridge University Press}
}

@article{aerts2010model,
  title={Model selection in regression based on pre-smoothing},
  author={Aerts, M. and Hens, N. and Simonoff, J. S.},
  journal={Journal of Applied Statistics},
  volume={37},
  number={9},
  pages={1455--1472},
  year={2010},
  publisher={Taylor \& Francis}
}

@book{anderson1958introduction,
  title={Introduction to Multivariate Statistical Analysis},
  author={Anderson, T. W.},
  year={2003},
  edition = {3rd},
  publisher={Wiley}
}

@article{anderson1997effects,
  title={Effects of gastropod grazers on recruitment and succession of an estuarine assemblage: a multivariate and univariate approach},
  author={Anderson, M. J. and Underwood, A. J.},
  journal={Oecologia},
  volume={109},
  pages={442--453},
  year={1997},
  publisher={Springer}
}

@article{bates2024cross,
  title={Cross-validation: what does it estimate and how well does it do it?},
  author={Bates, Stephen and Hastie, Trevor and Tibshirani, Robert},
  journal={Journal of the American Statistical Association},
  volume={119},
  number={546},
  pages={1434--1445},
  year={2024},
  publisher={Taylor \& Francis}
}

@book{billingsley2013convergence,
  title={Convergence of Probability Measures},
  author={Billingsley, P.},
  year={2013},
  edition={2nd},
  publisher={John Wiley \& Sons}
}

@article{bonnini2022relationship,
  title={Relationship between mental health and socio-economic, demographic and environmental factors in the {COVID}-19 lockdown period-a multivariate regression analysis},
  author={Bonnini, S. and Borghesi, M.},
  journal={Mathematics},
  volume={10},
  number={18},
  pages={3237},
  year={2022},
  publisher={MDPI}
}

@article{bura2008distribution,
  title={On the distribution of the left singular vectors of a random matrix and its applications},
  author={Bura, E. and Pfeiffer, R.},
  journal={Statistics \& Probability Letters},
  volume={78},
  number={15},
  pages={2275--2280},
  year={2008},
  publisher={Elsevier}
}

@article{burnham1999latent,
  title={Latent variable multivariate regression modeling},
  author={Burnham, A. J. and MacGregor, J. F. and Viveros, R.},
  journal={Chemometrics and Intelligent Laboratory Systems},
  volume={48},
  number={2},
  pages={167--180},
  year={1999},
  publisher={Elsevier}
}

@book{cramer1999mathematical, 
title={Mathematical
Methods of Statistics}, 
author={Cram{\'e}r, H.}, 
volume={26}, 
year={1999}, 
publisher={Princeton University Press} 
}

@article{cristobal1987class,
  title={A class of linear regression parameter estimators constructed by nonparametric estimation},
  author={{Cristobal Cristobal}, J. A. and {Faraldo Roca}, P. and {Gonzalez Manteiga}, W.},
  journal={The Annals of Statistics},
  volume={15},
  number={2},
  pages={603--609},
  year={1987},
  publisher={Institute of Mathematical Statistics}
}

@article{faraldo1987efficiency,
  title={On efficiency of a new class of linear regression estimates obtained by preliminary non-parametric estimation},
  author={{Faraldo Roca}, P. and {Gonz{\'a}lez Manteiga}, W.},
  journal={New Perspectives in Theoretical and Applied Statistics},
  pages={229--242},
  year={1987},
  publisher={John Wiley New York}
}

@article{ferraty2012presmoothing,
  title={Presmoothing in functional linear regression},
  author={Ferraty, F. and {Gonz{\'a}lez Manteiga}, W. and {Mart{\'\i}nez Calvo}, A. and Vieu, P.},
  journal={Statistica Sinica},
  pages={69--94},
  year={2012},
  publisher={JSTOR}
}

@article{franck2011effect,
  title={The effect of particle size on cardiovascular disorders-The smaller the worse},
  author={Franck, U. and Odeh, S. and Wiedensohler, A. and Wehner, B. and Herbarth, O.},
  journal={Science of the Total Environment},
  volume={409},
  number={20},
  pages={4217--4221},
  year={2011},
  publisher={Elsevier}
}

@article{gospodinov2017spurious,
  title={Spurious inference in reduced-rank asset-pricing models},
  author={Gospodinov, N. and Kan, R. and Robotti, C.},
  journal={Econometrica},
  volume={85},
  number={5},
  pages={1613--1628},
  year={2017},
  publisher={Wiley Online Library}
}

@book{greene2018econometric,
  title={Econometric Analysis},
  author={Greene, W. H.},
  year={2018},
  edition={8th},
  publisher={Pearson Education India}
}

@article{hitchcock2006improved,
  title={Improved estimation of dissimilarities by presmoothing functional data},
  author={Hitchcock, D. B. and Casella, G. and Booth, J. G.},
  journal={Journal of the American Statistical Association},
  volume={101},
  number={473},
  pages={211--222},
  year={2006},
  publisher={Taylor \& Francis}
}

@article{hitchcock2007effect,
  title={The effect of pre-smoothing functional data on cluster analysis},
  author={Hitchcock, D. B. and Booth, J. G. and Casella, G.},
  journal={Journal of Statistical Computation and Simulation},
  volume={77},
  number={12},
  pages={1043--1055},
  year={2007},
  publisher={Taylor \& Francis}
}

@article{izenman1975reduced,
  title={Reduced-rank regression for the multivariate linear model},
  author={Izenman, A. J.},
  journal={Journal of Multivariate Analysis},
  volume={5},
  number={2},
  pages={248--264},
  year={1975},
  publisher={Elsevier}
}

@book{izenman2008modern,
  title={Modern Multivariate Statistical Techniques},
  author={Izenman, A. J.},
  volume={1},
  year={2008},
  publisher={Springer}
}

@article{janssen2001efficiency,
  title={Efficiency of linear regression estimators based on presmoothing},
  author={Janssen, P. and Swanepoel, J. and Veraverbeke, N.},
  journal={Communications in Statistics-Theory and Methods},
  volume={30},
  number={10},
  pages={2079--2097},
  year={2001},
  publisher={Taylor \& Francis}
}

@book{johnson07:applied,
  title={Applied Multivariate Statistical Analysis},
  author={Johnson, R. A. and Wichern, D. W. and others},
  year={2007},
  edition = {6th},
  publisher={Prentice Hall Upper Saddle River, NJ}
}

@article{jolliffe2016principal,
  title={Principal component analysis: a review and recent developments},
  author={Jolliffe, I. T. and Cadima, J.},
  journal={Philosophical Transactions of the Royal Society A: Mathematical, Physical and Engineering Sciences},
  volume={374},
  number={2065},
  pages={20150202},
  year={2016},
  publisher={The Royal Society Publishing}
}

@article{kim2009multivariate,
  title={A multivariate regression approach to association analysis of a quantitative trait network},
  author={Kim, S. and Sohn, K.-A. and Xing, E. P.},
  journal={Bioinformatics},
  volume={25},
  number={12},
  pages={i204--i212},
  year={2009},
  publisher={Oxford University Press}
}

@article{liang2015assessing,
  title={Assessing Beijing's {PM2.5} pollution: severity, weather impact, {APEC} and winter heating},
  author={Liang, X. and Zou, T. and Guo, B. and Li, S. and Zhang, H. and Zhang, S. and Huang, H. and Chen, S. X.},
  journal={Proceedings of the Royal Society A: Mathematical, Physical and Engineering Sciences},
  volume={471},
  number={2182},
  pages={20150257},
  year={2015},
  publisher={The Royal Society Publishing}
}

@article{liu2014review,
  title={A review of multivariate analyses in imaging genetics},
  author={Liu, J. and Calhoun, V. D.},
  journal={Frontiers in Neuroinformatics},
  volume={8},
  pages={29},
  year={2014},
  publisher={Frontiers Media SA}
}

@article{ma2022multiple,
  title={Multiple Change Points Detection in High-Dimensional Multivariate Regression},
  author={Ma, X. and Zhou, Q. and Zi, X.},
  journal={Journal of Systems Science and Complexity},
  volume={35},
  number={6},
  pages={2278--2301},
  year={2022},
  publisher={Springer}
}

@book{mardia95,	
author = {Mardia, K. V. and Kent, J. T. and Bibby, J. M.},	
title = {Multivariate Analysis},	
publisher = {Academic Press},	
year = {1995},	
edition = {10th},
series = {Probability and Mathematical Statistics},
address = {San Diego}
}

@article{musta2022presmoothing,
  title={A presmoothing approach for estimation in the semiparametric {C}ox mixture cure model},
  author={Musta, E. and Patilea, V. and Van Keilegom, I.},
  journal={Bernoulli},
  volume={28},
  number={4},
  pages={2689--2715},
  year={2022},
  publisher={Bernoulli Society for Mathematical Statistics and Probability}
}

@article{musta2024regression,
  title={Regression estimation using surrogate responses obtained by presmoothing},
  author={Musta, E. and Patilea, V. and Van Keilegom, I.},
  journal={Statistica Neerlandica},
  year={2024},
  note = {(in press)},
  publisher={Wiley Online Library}
}

@article{neudecker1990asymptotic,
  title={The asymptotic variance matrix of the sample correlation matrix},
  author={Neudecker, H. and Wesselman, A. M.},
  journal={Linear Algebra and its Applications},
  volume={127},
  pages={589--599},
  year={1990},
  publisher={North-Holland}
}

@article{petrella2019joint,
  title={Joint estimation of conditional quantiles in multivariate linear regression models with an application to financial distress},
  author={Petrella, L. and Raponi, V.},
  journal={Journal of Multivariate Analysis},
  volume={173},
  pages={70--84},
  year={2019},
  publisher={Elsevier}
}

@article{pun2017long,
  title={Long-term {PM2.5} exposure and respiratory, cancer, and cardiovascular mortality in older {US} adults},
  author={Pun, V. C. and Kazemiparkouhi, F. and Manjourides, J. and Suh, H. H.},
  journal={American Journal of Epidemiology},
  volume={186},
  number={8},
  pages={961--969},
  year={2017},
  publisher={Oxford University Press}
}

@techreport{ratsimalahelo2001rank,
  title={Rank test based on matrix perturbation theory},
  author={Ratsimalahelo, Z.},
  year={2001},
  institution={EERI Research Paper Series}
}

@Manual{Rcore,
    title = {R: A Language and Environment for Statistical Computing},
    author = {{R Core Team}},
    organization = {R Foundation for Statistical Computing},
    address = {Vienna, Austria},
    year = {2024},
    url = {https://www.R-project.org/},
  }

@book{Reinsel1998,
   author = {Reinsel, G. and Velu, R. and Chen, K.},
   title = {Multivariate Reduced-Rank Regression: Theory and Applications},
   series = {Lecture Notes in Statistics},
   publisher = {Springer New York},
   year = {2022},
}

@article{sivakumaran2011abundant,
  title={Abundant pleiotropy in human complex diseases and traits},
  author={Sivakumaran, S. and Agakov, F. and Theodoratou, E. and Prendergast, J. G. and Zgaga, L. and Manolio, T. and Rudan, I. and McKeigue, P. and Wilson, J. F. and Campbell, H.},
  journal={The American Journal of Human Genetics},
  volume={89},
  number={5},
  pages={607--618},
  year={2011},
  publisher={Elsevier}
}

@article{skagerberg1992multivariate,
  title={Multivariate data analysis applied to low-density polyethylene reactors},
  author={Skagerberg, B. and MacGregor, J. F. and Kiparissides, C.},
  journal={Chemometrics and Intelligent Laboratory Systems},
  volume={14},
  number={1-3},
  pages={341--356},
  year={1992},
  publisher={Elsevier}
}

@article{srivastava2003predicting,
  title={Predicting multivariate response in linear regression model},
  author={Srivastava, M. S. and Solanky, T. K. S.},
  journal={Communications in Statistics-Simulation and Computation},
  volume={32},
  number={2},
  pages={389--409},
  year={2003},
  publisher={Taylor \& Francis}
}

@article{tedesco2025instrumental,
  title={Instrumental variable estimation of the proportional hazards model by presmoothing},
  author={Tedesco, L. and Beyhum, J. and Van Keilegom, I.},
  journal={Electronic Journal of Statistics},
  volume={19},
  number={1},
  pages={656--717},
  year={2025},
  publisher={The Institute of Mathematical Statistics and the Bernoulli Society}
}

@INPROCEEDINGS{tonks,
  author={Tonks, Z. and Sankaran, G. and Davenport, J.},
  booktitle={2017 19th International Symposium on Symbolic and Numeric Algorithms for Scientific Computing (SYNASC)}, 
  title={Fast Matrix Operations in Computer Algebra}, 
  year={2017},
  volume={},
  number={},
  pages={67-70}
}

@article{wille2004sparse,
  title={Sparse graphical {G}aussian modeling of the isoprenoid gene network in {A}rabidopsis thaliana},
  author={Wille, A. and Zimmermann, P. and Vranov{\'a}, E. and F{\"u}rholz, A. and Laule, O. and Bleuler, S. and Hennig, L. and Preli{\'c}, A. and von Rohr, P. and Thiele, L. and others},
  journal={Genome Biology},
  volume={5},
  pages={1--13},
  year={2004},
  publisher={Springer}
}

@article{zou2022estimation,
  title={Estimation of low rank high-dimensional multivariate linear models for multi-response data},
  author={Zou, C. and Ke, Y. and Zhang, W.},
  journal={Journal of the American Statistical Association},
  volume={117},
  number={538},
  pages={693--703},
  year={2022},
  publisher={Taylor \& Francis}
}

@article{massy1965principal,
  title={Principal components regression in exploratory statistical research},
  author={Massy, William F},
  journal={Journal of the American Statistical Association},
  volume={60},
  number={309},
  pages={234--256},
  year={1965},
  publisher={Taylor \& Francis}
}

@incollection{jolliffe2011principal,
  title={Principal component analysis},
  author={Jolliffe, Ian},
  booktitle={International encyclopedia of statistical science},
  pages={1094--1096},
  year={2011},
  publisher={Springer}
}

@inproceedings{james1961estimation,
  title={Estimation with quadratic loss},
  author={James, William and Stein, Charles and others},
  booktitle={Proceedings of the fourth Berkeley symposium on mathematical statistics and probability},
  volume={1},
  number={1961},
  pages={361--379},
  year={1961},
  organization={University of California Press}
}

@Book{mvtnorm2009,
  title = {Computation of Multivariate Normal and t Probabilities},
  author = {Alan Genz and Frank Bretz},
  series = {Lecture Notes in Statistics},
  year = {2009},
  publisher = {Springer-Verlag},
  address = {Heidelberg},
  isbn = {978-3-642-01688-2},
}

\newpage
\appendix
\section{Proof of theoretical results}
\label{app:proofs}

In this appendix, we provide details of the proof of Theorem \ref{theorem:lrpsdist} and Proposition \ref{prop:complexity} described in the main text.

\subsection{Proof of Theorem \ref{theorem:lrpsdist}}
\label{sec:appendix}


We begin by establishing a result which be useful in the proof of the theorem.  Let $Y$ be a random data matrix from model \eqref{eq:mrreg} (where the fourth order moments of the errors are finite).  Note that in what follows, we assume without loss of generality that $X$ and $Y$ are mean-centered, so that $S_y = n^{-1}Y^{\top}Y$ is  the sample covariance matrix, and define $\Sigma_{y}=n^{-1}(XB)^{\top}(XB)+\Sigma_{e}$.\\

\begin{lemma}
Let $S_y$ and $\Sigma_y$ be as above.  Then
\begin{equation}
\sqrt{n} \operatorname{vec}(S_y-\Sigma_y) \stackrel{D}{\rightarrow} \mathcal{N}_{q^2}\left(0, \mathbf{V}_y\right),
\end{equation}
where 
\begin{equation}\label{eq:Vy}
    \mathbf{V}_y=\mathbb{E}\left[e e^{\top} \otimes e e^{\top}\right]-\operatorname{vec}\left(\Sigma_e\right) \otimes \operatorname{vec}\left(\Sigma_e\right)^{\top}
\end{equation}
in which $e\in\mathbb{R}^{q\times1}$ denotes a generic row of $E$. 
\label{lemma1}
\end{lemma}
\begin{proof}

Denote by $S_{e}=n^{-1}E^{\top}E$ the sample covariance of the error vectors $E$ with corresponding population covariance matrix $\Sigma_e$. Since the fourth order moments of the error distribution are assumed finite and the rows of $E$ are assumed to be independent and identically distributed under model \eqref{eq:mrreg}, Theorem 1 in \cite{neudecker1990asymptotic} (with zero mean random variables) establishes that $S_e$ is a root-$n$ consistent estimator for $\Sigma_e$:
\begin{equation} \label{eq:asyvarSe}
\sqrt{n} \operatorname{vec}(S_e-\Sigma_e) \stackrel{D}{\rightarrow} \mathcal{N}_{q^2}\left(0, \mathbf{V}_e\right),
\end{equation}
where 
\begin{equation*}
    \mathbf{V}_e = \mathbb{E}\left[e e^{\top} \otimes e e^{\top}\right] - \operatorname{vec} (\Sigma_e) \otimes \operatorname{vec} (\Sigma_e)^{\top}
\end{equation*}
and $e$ is as above.\\

Let $Z := n^{-1}(XB)^{\top}(XB) = n^{-1}B^{\top}(X^{\top}X)B \in\mathbb{R}^{q\times q}$ and $W := n^{-1}(XB)^{\top}E \in\mathbb{R}^{q\times q}$.

Using the linear model \eqref{eq:mrreg}, the sample covariance matrix of $Y$, $S_y$, can be expressed as
\begin{equation}\label{eq:scovexp}
S_{y} =n^{-1}Y^{\top}Y =n^{-1}\Bigl((XB)^{\top}(XB)+(XB)^{\top}E+E^{\top}(XB)+E^{\top}E\Bigr).
\end{equation}

Then $$\Sigma_{y}:=\mathbb{E}(S_y)
  =n^{-1}\Bigl((XB)^{\top}(XB)\Bigr)+n^{-1}\mathbb{E}\Bigl((XB)^{\top}E\Bigr)+n^{-1}\mathbb{E}\Bigl(E^{\top}(XB)\Bigr)+\Sigma_{e}.$$

Because $\mathbb{E}(E)=0$ (and thus $\mathbb{E}(W)=0$), this simplifies to 
\begin{equation}\label{eq:Sigy}
\Sigma_{y}=n^{-1}(XB)^{\top}(XB)+\Sigma_{e}.
\end{equation}

Then combining \eqref{eq:scovexp} and \eqref{eq:Sigy}, we get
\begin{eqnarray}\nonumber
\underbrace{\sqrt{n}\operatorname{vec}(S_{y}-\Sigma_{y})}_{C} 
& = &\sqrt{n}\operatorname{vec}\left((Z + W + W^{\top} + S_e) - \Sigma_y\right)\\\nonumber
& = &\sqrt{n}\operatorname{vec}\left(( Z + W + W^{\top} + S_e) - \left(Z + \Sigma_e\right)\right)\\
& = & \underbrace{\sqrt{n}\operatorname{vec}\left( 
( W + W^{\top})\right)}_{(A_1)}+\underbrace{\sqrt{n}\operatorname{vec}\left(S_e - \Sigma_e \right)}_{(A_2)}.\label{eq:Wrem}
\end{eqnarray}

Using the properties of the vec operator \citep{abadir},
$$
\operatorname{vec}\left(W\right) = \operatorname{vec}\left(n^{-1}\nonumber(XB)^{\top}E \right) = \underbrace{\left( I_q \otimes n^{-1}(XB)^{\top} \right)}_A \operatorname{vec}(E),$$

where $\otimes$ denotes the Kronecker product, so
$$
\begin{aligned}
\operatorname{Var} \left(\operatorname{vec}(W) \right) & =\operatorname{Var} \left(A \operatorname{vec}(E)\right) = A \operatorname{Var}\left(\operatorname{vec}(E)\right)A^{\top} \\
&= \left( I_q \otimes n^{-1}(XB)^{\top} \right) (\Sigma_e \otimes I_n)\left( I_q \otimes n^{-1}(XB)^{\top} \right)^{\top}\\
&= \left( I_q \Sigma_e I_q \right)\otimes\left(n^{-2} (XB)^{\top} I_n (XB)\right)\\
&= \Sigma_e \otimes n^{-1}Z,
\end{aligned}
$$
using the transpose and mixed-product properties of the Kronecker product.  Assuming that  (with mean-centered $X$) \ $\lim_{n\rightarrow\infty}S_{X}=\Sigma_{x}$ (where $S_{X}=n^{-1}X^{\top}X$) for some fixed positive definite matrix $\Sigma_x$ (as is customary in regression), 
we have
$$
\begin{aligned}
\mathbb{E}\left(\operatorname{vec}(W)\right) &= 0 \quad \textrm{and}\\
\operatorname{Var} \left(\operatorname{vec}(W) \right) & =\Sigma_e \otimes n^{-1}Z =\left(\Sigma_e \otimes n^{-1}B^{\top} S_{X}B\right) \rightarrow 0.
\end{aligned}
$$
Hence considering the first term $A_1$ in \eqref{eq:Wrem}, we have
$$\operatorname{vec}\left(( W + W^{\top})\right)=\left(\operatorname{vec}(W)\right)+\left(\operatorname{vec}(W^{\top})\right)\stackrel{\mathbb{P}}{\rightarrow} 0$$


thus considering \eqref{eq:Wrem} and using (the multivariate version of) Slutsky's theorem (see e.g., \cite[
Section 20.6]{cramer1999mathematical}), we have

$$
\begin{aligned}
\sqrt{n}\operatorname{vec}(S_{y}-\Sigma_{y})  \stackrel{D}{\rightarrow} \mathcal{N}_{q^2 }(0,\mathbf{V}_y),
\end{aligned}
$$
with 
   $$ \mathbf{V}_y=\mathbb{E}\left[e e^{\top} \otimes e e^{\top}\right]-\operatorname{vec}\left(\Sigma_e\right) \otimes \operatorname{vec}\left(\Sigma_e\right)^{\top}.$$
\end{proof}

\bigskip
We now present a theorem and a corollary without proof, both taken from \cite{bura2008distribution}, which are used as intermediary results to establish Theorem \ref{theorem:lrpsdist}.

\begin{unntheorem}
Let $\hat{\Omega}$ be a $p\times q$-dimensional random matrix that is a function of a random sample of $p$-dimensional data ${Z_1,\dots,Z_n}$. Assume that $\hat{\Omega}_n$ is asymptotically normally distributed as 
$$\sqrt{n}\operatorname{vec}(\hat{\Omega}_n-\Omega)  \stackrel{D}{\rightarrow} \mathcal{N}_{pq}(0,\mathbf{V}).$$

Let $r = rank(\Omega)$ with $d \leq \min(p,q)$, and let the singular value decomposition of $\Omega$ be 
$$\Omega=
L^{\top}\left(\begin{array}{ll}
D & 0 \\
\mathbf{0} & 0
\end{array}\right) R,
$$
where $L^{\top}=\left(L_1, L_0\right)$ is of order $p \times p$ with $L_1: p \times r$, $L_0: p \times(p-r)$. $R^{\top}=\left(R_1, R_0\right)$ is orthogonal with $R_1: q \times r, R_0: q \times(q-r)$. The $r$ left singular vectors $U_1$ of $\Omega$  correspond to its $r$ non-zero singular values, $d_1\geq \dots\geq d_r$. Analagously, the SVD of $\hat{\Omega}$ is
\begin{equation}\label{eq:omhat}
\hat{\Omega}=\hat{L}^{\top}\left(\begin{array}{cc}\hat{D}_1 & \mathbf{0} \\
& \\
\mathbf{0} & \hat{D}_0
\end{array}\right) \hat{R},
\end{equation}
with $\hat{L}^{\top}=\left(\hat{L}_1, \hat{L}_0\right), \hat{R}^{\top}=\left(\hat{R}_1, \hat{R}_0\right)$, and the partition conforms with the partition of $\Omega$.

Assume that $\lambda_r > \epsilon > 0$, i.e.\ the minimum positive singular value of $\Omega$ is well separated from zero. Then as $n\rightarrow \infty$
\begin{enumerate}[(a)]
\item $\hat{D}_1 \stackrel{\mathbb{P}}{\rightarrow} D_1$ and $\hat{D}_0 \stackrel{\mathbb{P}}{\rightarrow} 0$.
\item $\hat{R}_1 \stackrel{\mathbb{P}}{\rightarrow} R_1$ and $\hat{L}_1 \stackrel{\mathbb{P}}{\rightarrow} L_1$.
\end{enumerate}\label{BPtheorem}
\end{unntheorem}

\begin{unncoro}
The matrix $\hat{L}_1$ of left singular vectors in \eqref{eq:omhat} is asymptotically normal with
\begin{equation}\label{eq:BPasydist}
n^{1/2}\operatorname{vec}(\hat{L}_1-L_1)\stackrel{D}{\rightarrow} N_{pd}\left(\mathbf{0}, \mathbf{\Sigma}_L\right), 
\end{equation}
where 
\begin{equation}\label{eq:BPasyvar}
\mathbf{\Sigma}_L =\left(D^{-1} R_1^{\top} \otimes I_p\right) \mathbf{V}\left(R_1 D^{-1} \otimes I_p\right).
\end{equation}
\end{unncoro}


The assumption on the minimum positive singular value ensures that $\Lambda^{-1}$ is bounded, necessary for the proof of the theorem (see \cite{bura2008distribution}).

\bigskip

We now apply the theorem and corollary to our setting.  

\bigskip

Let $V \Lambda V^{\top}$ be the compact form SVD of the quantity $\Sigma_y$ from Lemma \ref{lemma1} (equivalent to its eigendecomposition since it is positive semi-definite), and suppose $\Sigma_y$ has rank $r$. Suppose also that the random matrix $Y$ has the compact form decomposition $Y = \hat{U} \hat{D} \hat{V}^{\top}$, where the diagonal entries of $\hat{D}$ are $\{\hat{d}_1,\dots,\hat{d}_r\}$. Then
$S_y$ can be written as 
$$
    n^{-1}Y^{\top}Y = n^{-1} \hat{V} \hat{D} \hat{U}^{\top} \hat{V} \hat{D} \hat{U}^{\top} \\
     = \hat{V} \hat{\Lambda} \hat{V}^{\top}, 
$$
due to the orthogonality of $\hat{U}$, and where $\hat{\Lambda}=\operatorname{diag}\left(\hat{\lambda}_1, \hat{\lambda}_2, \ldots, \hat{\lambda}_r\right)$ is a diagonal matrix containing the descending eigenvalues $\hat\lambda_i=\hat d_i^2 /n, i=1,\dots,r$.  

Now using the result established in Lemma \ref{lemma1}, the asymptotic normality assumption of the theorem is satisfied for $S_y$ as an estimator of $\Sigma_y$\footnote{Note that the setup for the theorem in \cite{bura2008distribution} assumes independent and identically distributed data ${Z_1,\dots,Z_n}$.  However, the arguments in the proof, predominantly borrowed from \cite{ratsimalahelo2001rank}, only require the root-$n$ consistency of $\hat{\Omega}$ as an estimator for $\Omega$.  Thus the theorem applies under the more general setting of independent, but non-identically distributed data as in the regression model \eqref{eq:mrreg}.}, hence $\hat{\Lambda} \stackrel{\mathbb{P}}{\rightarrow}\Lambda$, and the corollary gives that
the matrix $\hat{V}$ of eigenvectors of $S_y$ are asymptotically normal with
\begin{equation*}
n^{1 / 2} \operatorname{vec}\left(\hat{V}-V\right) \stackrel{D}{\rightarrow} N_{qr}\left(\mathbf{0}, \mathbf{\Sigma}_{\hat{V}}\right)
\quad \textrm{where} 
\quad
\mathbf{\Sigma}_{\hat{V}} =\left(\Lambda^{-1} V^{\top} \otimes I_q\right) \mathbf{V}_y \left(V \Lambda^{-1} \otimes I_q\right).
\end{equation*}

Hence considering the subset of first $k$ eigenvectors and corresponding eigenvalues ($k<r$), we have that
\begin{equation}\label{eq:asydistvk}
n^{1 / 2} \operatorname{vec}\left(\hat{V}_k-V_k\right) \stackrel{D}{\rightarrow} N_{qk}\left(\mathbf{0}, \mathbf{\Sigma}_{\hat{V}_k}\right),
\end{equation}
with
\begin{equation*}\label{eq:asyvarvk}
\mathbf{\Sigma}_{\hat{V}_k} =\left({\Lambda_k}^{-1} {V_k}^{\top} \otimes I_q\right) \mathbf{V}_y \left(V_k {\Lambda_k}^{-1} \otimes I_q\right),
\end{equation*}
where $\hat V_k$, $V_k$ and $\Lambda_k$ correspond to the quantities above, curtailed to $k$ eigenvectors / eigenvalues.

\bigskip\bigskip

In view of the relationship between the LRPS estimator and the OLS estimator $\hat{B}\hat{V}_k{\hat{V}_k}^{\top}$, we have 
$
\operatorname{vec}\left(\hat{B}\hat{V}_k{\hat{V}_k}^{\top}\right) = \left(\hat{V}_k{\hat{V}_k}^{\top}\otimes I_p \right) \operatorname{vec}(\hat{B})$, and the corresponding expression with the true eigenvectors also holds.\\

Now considering the quantity $\operatorname {vec}\left(\hat{V}_k{\hat{V}_k}^{\top}\right)$, due to the consistency result for $\hat{V_k}$ in \eqref{eq:asydistvk}, the continuous mapping theorem (see e.g. \cite{billingsley2013convergence}) then means that $\hat{V}_k{\hat{V}_k}^{\top}$ is a consistent estimator of $V_k V_k^{\top}$, i.e. $\hat{V}_k{\hat{V}_k}^{\top} \stackrel{\mathbb{P}}{\rightarrow} V_k V_k^{\top}$ (see also the proof of Theorem 1 in \cite{bura2008distribution}). Using the distributional property \eqref{eq:olsdist} and applying Slutsky's theorem gives that 



\begin{eqnarray*}
    n^{1 / 2} \operatorname{vec}\left(\tilde{B}- BV_k{V_k}^{\top}\right) &=& n^{1 / 2} \operatorname{vec}\left(\hat{B}(\hat{V}_k{\hat{V}_k}^{\top}-V_k{V_k}^{\top})+(\hat{B}- B)V_k{V_k}^{\top}\right)\\
    &\stackrel{D}{\rightarrow}& \mathcal{N}\left(\mathbf{0}, V_kV_k^{\top} \Sigma_e V_kV_k^{\top} \otimes S_X^{-1} \right),
\end{eqnarray*}

where the variance term is derived as 
\begin{eqnarray*}
\operatorname{Var}\left(\left(V_k V_k^{\top}\otimes I_p\right)\operatorname{vec}(n^{1/2}\hat{B})\right)&=&\left(V_k V_k^{\top}\otimes I_p\right)n\operatorname{Var}\left(\operatorname{vec}(\hat{B})\right )\left(V_k V_k^{\top}\otimes I_p\right)^{\top}\\
&=&\left(V_k V_k^{\top}\otimes I_p\right)\left(\Sigma_e \otimes S_X^{-1}\right)\left(V_k V_k^{\top}\otimes I_p\right)\\
&=&V_kV_k^{\top} \Sigma_e V_kV_k^{\top} \otimes S_X^{-1},
\end{eqnarray*}
using the mixed-product property of Kronecker products.\\

Note that this result will also be valid if the distributional result \eqref{eq:olsdist} is asymptotic, as in the case when the error distribution $G$ in the regression model \eqref{eq:mrreg} is not Gaussian.

\subsection{Proof of Proposition \ref{prop:complexity}}
\label{sec:appendixprop}



The optimal computational complexity of the LRPS estimator will be different depending on whether $q > p$ or $p > q$, as seen below.

\subsubsection*{The case when $p > q$}
The LRPS estimator can be computed as follows:
$$
\tilde{B}
=
\underbrace{\underbrace{\underbrace{(X^{\top} X)^{-1}}_{O(np^2+p^3)}
\;\;
\underbrace{X^{\top} U_k}_{O(npk)}}_{O(p^2k)}
\;\;
\underbrace{D_k V_k^{\top}}_{O(qk)}}_{O(pqk)},
$$
After performing the truncated SVD of $Y$, the operations required to compute the LRPS estimator are: projection, low-rank reconstruction, and two matrix multiplications, contributing $O(npk)$, $O(qk)$, $O(p^2k)$ and $O(pqk)$ respectively, where we note that the matrix multiplication with $D_k$ can be done efficiently since it is diagonal.  This leads to a total complexity of
$$
O(np^2+p^3) + O(nqk) + O(npk) + O(qk) + O(p^2k) + O(pqk).
$$

\subsubsection*{The case when $q > p$}
Similarly, in this case the LRPS estimator can be calculated as:
$$
\begin{aligned}
\tilde{B} &=
\underbrace{\underbrace{\underbrace{\underbrace{(X^{\top} X)^{-1}}_{O(np^2+p^3)}
\;\;
\underbrace{X^{\top} U_k}_{O(npk)}}_{O(p^2k)}
\;\;
D_k}_{O(pk)} V_k^{\top}}_{O(pqk)} , 
\end{aligned}
$$
leading to a total complexity of 
$$
O(np^2+p^3) + O(nqk) + O(npk) + O(pk) + O(p^2k) + O(pqk).
$$

\section{Additional computational comparisons}\label{app:runtime}
In this section, we present additional simulations to assess runtimes and memory usage for different multiresponse regression methods as outlined in the main text. Below we report the wall clock computation times for fixed $q$ and fixed $p$ (Tables \ref{tab:runtimep} and \ref{tab:runtimeq}) and the memory consumption (Tables \ref{tab:memoryp} and \ref{tab:memoryq}), where $S1-S5$ indicates a scaling factor to increase the size of the regression problem by the dimension ($n$, $p$ or $q$) under consideration (see Section \ref{sec:compcost}).      

\begin{table}[!h]
\centering
\caption{Average runtime (sec) for regression methods described in the text for different dimension sizes. Left: Increasing $n$; Right: Increasing $p$.}
\begin{tabular}{lccccccccccc}
\toprule
& \multicolumn{5}{c}{Increasing $n$ runtime ($\times 10^{-3}$s)} & & \multicolumn{5}{c}{Increasing $p$ runtime ($\times 10^{-3}$s)} \\
\cmidrule{2-6} \cmidrule{8-12}
Method & S1 & S2 & S3 & S4 & S5 & & S1 & S2 & S3 & S4 & S5 \\
\midrule
OLS  & 0.15 & 0.24 & 0.27 & 0.44 & 0.59 & 
& 0.15 & 0.31 & 0.58 & 0.99 & 1.59 \\

RRR  & 0.27 & 0.47 & 0.58 & 0.79 & 0.91 &
& 0.25 & 0.41 & 0.75 & 1.12 & 1.69 \\

LRPS & 0.69 & 0.36 & 0.39 & 0.48 & 0.54 &
& 0.23 & 0.34 & 0.98 & 0.94 & 1.88 \\

PCR  & 0.46 & 0.71 & 0.87 & 1.42 & 1.25 &
& 0.35 & 0.76 & 1.76 & 1.34 & 1.94 \\
\bottomrule
\end{tabular}\label{tab:runtimep}
\end{table}

\begin{table}[!h]
\centering
\caption{Average runtime (sec) for regression methods described in the text for for different dimension sizes. Left: Increasing $n$; Right: Increasing $q$.}
\begin{tabular}{lccccccccccc}
\toprule
& \multicolumn{5}{c}{Increasing $n$ runtime ($\times 10^{-3}$s)} & & \multicolumn{5}{c}{Increasing $q$ runtime ($\times 10^{-3}$s)} \\
\cmidrule{2-6} \cmidrule{8-12}
Method & S1 & S2 & S3 & S4 & S5 & & S1 & S2 & S3 & S4 & S5 \\
\midrule
OLS  & 0.27 & 0.40 & 0.61 & 0.84 & 1.04 &
& 0.24 & 0.28 & 0.35 & 0.64 & 0.71 \\

RRR  & 0.70 & 1.06 & 1.59 & 2.14 & 2.71 &
& 0.62 & 1.33 & 2.11 & 3.17 & 4.46 \\

LRPS & 0.47 & 0.68 & 1.06 & 1.58 & 1.92 &
& 0.33 & 0.65 & 0.84 & 1.07 & 1.30 \\

PCR  & 0.80 & 0.51 & 0.62 & 0.60 & 0.83 &
& 0.37 & 0.37 & 0.29 & 0.33 & 0.39 \\
\bottomrule
\end{tabular}\label{tab:runtimeq}
\end{table}

\begin{table}[!h]
\centering
\caption{Average memory usage (MB) for multiresponse regression methods described in the text for different dimension sizes. Left: Increasing $n$; Right: Increasing $p$.}
\begin{tabular}{lccccccccccc}
\toprule
& \multicolumn{5}{c}{Increasing $n$ memory (kB)} & & \multicolumn{5}{c}{Increasing $p$ memory (kB)} \\
\cmidrule{2-6} \cmidrule{8-12}
Method & S1 & S2 & S3 & S4 & S5 & & S1 & S2 & S3 & S4 & S5 \\
\midrule
OLS  & 1.28 & 0.70 & 0.70 & 0.70 & 0.70 &
& 0.72 & 0.72 & 0.72 & 0.72 & 0.72 \\

RRR  & 0.70 & 0.70 & 0.70 & 0.70 & 0.70 &
& 0.70 & 0.70 & 0.70 & 0.70 & 0.70 \\

LRPS & 0.70 & 0.70 & 0.70 & 0.70 & 0.70 &
& 0.70 & 0.70 & 0.70 & 0.70 & 0.70 \\

PCR  & 0.70 & 0.70 & 0.70 & 0.70 & 0.70 &
& 0.70 & 0.70 & 0.70 & 0.70 & 0.70 \\
\bottomrule
\end{tabular}\label{tab:memoryp}
\end{table}

\begin{table}[!h]
\centering
\caption{Average memory usage (MB) for multiresponse regression methods described in the text for different dimension sizes. Left: Increasing $n$; Right: Increasing $q$.}
\begin{tabular}{lccccccccccc}
\toprule
& \multicolumn{5}{c}{Increasing $n$ memory (kB)} & & \multicolumn{5}{c}{Increasing $q$ memory (kB)} \\
\cmidrule{2-6} \cmidrule{8-12}
Method & S1 & S2 & S3 & S4 & S5 & & S1 & S2 & S3 & S4 & S5 \\
\midrule
OLS  & 0.93 & 0.70 & 0.70 & 0.70 & 0.70 &
& 0.93 & 0.93 & 0.93 & 0.93 & 0.93 \\

RRR  & 0.70 & 0.70 & 0.70 & 0.70 & 0.70 &
& 0.70 & 0.70 & 0.70 & 0.70 & 0.70 \\

LRPS & 0.70 & 0.70 & 0.70 & 0.70 & 0.70 &
& 0.70 & 0.70 & 0.70 & 0.70 & 0.70 \\

PCR  & 0.70 & 0.70 & 0.70 & 0.70 & 0.70 &
& 0.70 & 0.70 & 0.70 & 0.70 & 0.70 \\
\bottomrule
\end{tabular}\label{tab:memoryq}
\end{table}

\end{document}